\documentclass[11pt]{article}
\usepackage{amsmath, amssymb}
\usepackage{amsthm}
\usepackage{mathtools}
\usepackage[margin=1in]{geometry}
\usepackage{enumitem}
\usepackage{mathpazo}
\usepackage{centernot}
\usepackage{complexity}
\usepackage{hyperref}
\usepackage{tikz}
\usetikzlibrary{arrows,automata}
\usepackage{algorithm}
\usepackage[noend]{algpseudocode}

\makeatletter
\newtheorem*{rep@theorem}{\rep@title}
\newcommand{\newreptheorem}[2]{%
	\newenvironment{rep#1}[1]{%
		\def\rep@title{#2 \ref{##1}}%
		\begin{rep@theorem}}%
		{\end{rep@theorem}}}
\makeatother

\theoremstyle{plain}
\newtheorem{theorem}{Theorem}
\newreptheorem{theorem}{Theorem}

\newtheorem{lemma}[theorem]{Lemma}
\newreptheorem{lemma}{Lemma}
\newtheorem{proposition}[theorem]{Proposition}
\newtheorem{corollary}[theorem]{Corollary}

\newtheorem{problem}[theorem]{Problem} 
\newtheorem{property}[theorem]{Property} 

\theoremstyle{definition}
\newtheorem{definition}[theorem]{Definition}

\newcommand{\approxdeg}{\widetilde{\textrm{deg}}}

\newcommand{\eps}{\varepsilon}

\DeclareMathOperator{\bs}{bs}

\renewcommand{\lparen}{\texttt{(}}
\renewcommand{\rparen}{\texttt{)}}
\newcommand{\CFED}{\textsc{cf-ed}}

\newcommand{\adv}{\operatorname{Adv^\pm}}

\newcommand{\binL}{L^{\normalfont{\texttt{bin}}}}
\newcommand{\hashtag}{\#}

\begin{document}

\title{A Quantum Query Complexity Trichotomy for Regular Languages}
\author{
Scott Aaronson\thanks{Supported by a Vannevar Bush Fellowship from the US Department of Defense, a Simons Investigator Award, and the Simons "It from Qubit" collaboration.
} \\ UT Austin \\ \texttt{aaronson@utexas.edu }  \and 
Daniel Grier\thanks{Supported by an NSF Graduate Research Fellowship under Grant No. 1122374.} \\ MIT \\ \texttt{grierd@mit.edu} \and
Luke Schaeffer \\ MIT \\ \texttt{lrs@mit.edu}
}
\date{}
\maketitle
%!TEX root = ../reg-query.tex

\begin{abstract}
We present a trichotomy theorem for the quantum query complexity of regular languages. Every regular language has quantum query complexity $\Theta(1)$, $\tilde{\Theta}(\sqrt n)$, or $\Theta(n)$. The extreme uniformity of regular languages prevents them from taking any other asymptotic complexity. This is in contrast to even the context-free languages, which we show can have query complexity $\Theta(n^c)$ for all computable $c \in [1/2,1]$. Our result implies an equivalent trichotomy for the approximate degree of regular languages, and a dichotomy---either $\Theta(1)$ or $\Theta(n)$---for sensitivity, block sensitivity, certificate complexity, deterministic query complexity, and randomized query complexity.

The heart of the classification theorem is an explicit quantum algorithm which decides membership in any star-free language in $\tilde{O}(\sqrt n)$ time. This well-studied family of the regular languages admits many interesting characterizations, for instance, as those languages expressible as sentences in first-order logic over the natural numbers with the less-than relation. Therefore, not only do the star-free languages capture functions such as OR, they can also express functions such as ``there exist a pair of 2's such that everything between them is a 0."  

Thus, we view the algorithm for star-free languages as a nontrivial generalization of Grover's algorithm which extends the quantum quadratic speedup to a much wider range of string-processing algorithms than was previously known.  We show a variety of applications---new quantum algorithms for dynamic constant-depth Boolean formulas, balanced parentheses nested constantly many levels deep, binary addition, a restricted word break problem, and path-discovery in narrow grids---all obtained as immediate consequences of our classification theorem.
\end{abstract}

%!TEX root = ../reg-query.tex

\section{Introduction}
\label{sec:intro}

Regular languages have a long history of study in classical theoretical computer science, going back to Kleene in the 1950s \cite{kleene:1956}. The definition is extremely robust: there are many equivalent characterizations ranging from machine models (e.g., deterministic or non-deterministic finite automata, $o(\log \log n)$-space Turing machines \cite{stearns:1965}), to grammars (e.g., regular expressions, prefix grammars), to algebraic structures (e.g., recognition via monoids, the syntactic congruence, or rational series). 
%Regular languages are also very well-behaved in the sense that most natural questions are decidable (e.g., is the language infinite?), and most natural operations on languages (e.g., union, complement) are computable. 
Regular languages are closed under most natural operations (e.g., union, complement), and also most natural questions are decidable (e.g., is the language infinite?).  Perhaps for this reason, regular languages are also a useful pedagogical tool, serving as a toy model for theory of computation students to cut their teeth on.

We liken regular languages to the symmetric\footnote{A symmetric Boolean function $f \colon \{0,1\}^n \to \{0,1\}$ is such that the value of $f$ only depends on the Hamming weight of the input.} Boolean functions. That is, both are a restricted, (usually) tractable special case of a much more general object, and often the common thread between a number of interesting examples. We suggest that these special cases should be studied and thoroughly understood first, to test proof techniques, to make conjectures, and to gain familiarity with the setting. 

In this work, we hope to understand the regular languages from the lens of another great innovation of theoretical computer science---query complexity, particularly quantum query complexity.   Not only is query complexity one of the few models in which provable lower bounds are possible, but is also often the case that efficient algorithms actually achieve the query lower bound.  In this case, the query lower bound suggests an algorithm which was otherwise thought not to exist, as was famously the case for Grover's search algorithm.

In the case of query complexity, symmetric functions are extremely well-understood with complete characterizations known for deterministic, randomized, and quantum algorithms in both the zero-error and bounded-error settings \cite{beals:2001}.  However, to the authors' knowledge, regular languages have not been studied in the query complexity model despite the fact that they appear frequently in query-theoretic applications. 

For example, consider the OR function over Boolean strings.  This corresponds to deciding membership in the language recognized by the regular expression $(0 | 1)^{*} 1 (0 | 1)^{*}$. Similarly, the parity function is just membership in the regular language $(0^* 1 0^* 1)^{*} 0^*$. It is well known that the quantum query complexity of OR is $\Theta(\sqrt{n})$, whereas parity is known to require $\Theta(n)$ quantum queries. Yet, there is a two-state deterministic finite automaton for each language.  This raises the question:  what is the difference between these two languages that causes the dramatic discrepancy between their quantum query complexities? More generally, can we decide the quantum query complexity of a regular language given a description of the machine recognizing it? Are all quantum query complexities even possible? We answer all of these questions in this paper.

The main contribution of this work is the complete characterization of the quantum query complexity of regular languages (up to some technical details), manifest as the following trichotomy: every regular language has quantum query complexity $\Theta(1)$, $\tilde{\Theta}(\sqrt{n})$, or $\Theta(n)$. In the process, we get an identical trichotomy for approximate degree, and dichotomies---in this case, $\Theta(1)$ or $\Theta(n)$---for a host of other complexity measures including deterministic complexity, randomized query complexity, sensitivity, block sensitivity, and certificate complexity.

Many of the canonical examples of regular languages fall easily into one of the three categories via well-studied algorithms or lower bounds. For example, the upper bound for the OR function results from Grover's famous search algorithm, and the lower bounds for OR and parity functions are straightforward applications of either the polynomial method \cite{beals:2001} or adversary method \cite{ambainis:2002}. Nevertheless, it turns out that there exists a vast class of regular languages which have neither a trivial $\Omega(n)$ lower bound nor an obvious $o(n)$ upper bound resulting from a straightforward application of Grover's algorithm.  A central challenge of the trichotomy theorem for quantum query complexity was showing that these languages \emph{do} actually admit a quadratic quantum speedup.

One such example is the language $\Sigma^{*} (2 0^{*} 2) \Sigma^{*}$, where $\Sigma = \{ 0, 1, 2 \}$.  Although there is no \emph{finite} witness for the language (e.g., to find by Grover search), we show that it nevertheless has an $\tilde{O}(\sqrt n)$ quantum algorithm.  More generally, this language belongs to a subfamily of regular languages known as \emph{star-free languages} because they have regular expressions which avoid Kleene star (albeit with the addition of the complement operation).\footnote{For example, the star-free expression for $\Sigma^{*} (2 0^{*} 2) \Sigma^{*}$ is $\overline{\varnothing} 2 \overline{\overline{\varnothing} \{1,2\} \overline{\varnothing}} 2 \overline{\varnothing}$.} Like regular languages, the star-free languages have many equivalent characterizations: counter-free automata \cite{mcnaughton:1971}, predicates expressible in either linear temporal logic or first-order logic \cite{kamp:1968, mcnaughton:1971}, the preimages of finite aperiodic monoids \cite{schutzenberger:1965}, or cascades of reset automata \cite{krohn:1965}. The star-free languages are those regular languages which can be decided in $\tilde{O}(\sqrt{n})$ queries. As a result, reducing a problem to any one of the myriad equivalent representations of these languages yields a quadratic quantum speedup for that problem. 

Let us take McNaughton's characterization of star-free languages in first-order logic as one example \cite{mcnaughton:1971}.  That is, every star-free language can be expressed as a sentence in first-order logic over the natural numbers with the less-than relation and predicates $\pi_a$ for $a \in \Sigma$, such that $\pi_a(i)$ is true if input symbol $x_i$ is $a$.  We can easily express the OR function as $\exists i \;  \pi_1(i)$, or the more complicated language $\Sigma^{*} (2 0^{*} 2) \Sigma^{*}$ as
$$
\exists i \, \exists k \, \forall j \ \ \  i < k \wedge \pi_2(i) \wedge \pi_2(k) \wedge (i < j < k \implies \pi_0(j)).
$$
Our result gives an algorithm for this sentence and arbitrarily complex sentences like it. We see this as a far-reaching generalization of Grover's algorithm, which extends the Grover speedup to a much wider range of string processing problems than was previously known.\footnote{Readers familiar with descriptive complexity will recall that $\AC^{0}$ has a similar, but somewhat more general characterization in first-order logic. It follows that all star-free languages, which have quantum query complexity $\tilde{O}(\sqrt{n})$, are in $\AC^{0}$. Conversely, we will show that regular languages \emph{not} in $\AC^{0}$ have quantum query complexity $\Omega(n)$. Thus, another way to state the trichotomy is that \emph{very roughly speaking} regular languages in $\NC^{0}$ have complexity $O(1)$, regular languages in $\AC^{0}$ but not $\NC^0$ have complexity $\tilde{\Theta}(\sqrt{n})$, and everything else has complexity $\Omega(n)$.}

%Another perspective comes from the characterization of regular languages as constant-space Turing machines. So, we can take a simple Turing machine and restrict it to instances which only require constant space. For example, the \emph{Dyck language} is the language of all strings over parentheses (i.e., $\texttt{(}$ and $\texttt{)}$) such that the parentheses are balanced. In general, this language is context free, but let us

\subsection{Results}

Our main result is the following:
\begin{theorem}[informal]
\label{thm:classification}
Every\footnote{There are two caveats: the quantum query complexity may oscillate between asymptotically different functions; the quantum query complexity may also be zero.  For the formal statement of this theorem see Section~\ref{sec:formalstatement}. } regular language has quantum query complexity $\Theta(1)$, $\tilde{\Theta}(\sqrt{n})$, or $\Theta(n)$. Moreover, the quantum time complexity of each language matches its query complexity.
\end{theorem} 

The theorem and its proof have several consequences which we highlight below.
\begin{enumerate}
\item \textbf{Algebraic characterization:} We give a characterization of each class of regular languages in terms of the monoids that recognize them.  That is, the monoid is either a rectangular band, aperiodic, or finite.  In particular, given a description of the machine, grammar, etc. generating the language, we can decide its membership in one of the three classes by explicitly calculating its syntactic monoid and checking a small number of conditions.  See Section~\ref{sec:formalstatement}.
\item \textbf{Related complexity measures:} Many of the lower bounds are derived from lower bounds on other query measures.  To this end, we prove query dichotomies for deterministic complexity, randomized query complexity, sensitivity, block sensitivity, and certificate complexity---they are all either $\Theta(1)$ or $\Theta(n)$ for regular languages.  By standard relationships between the measures, this shows that approximate degree and quantum query complexity are either $O(1)$ or $\Omega(\sqrt{n})$.  See Section~\ref{sec:lowerbounds}.
\item \textbf{Generalization of Grover's algorithm:}  The $\BQP$ algorithm using $\tilde{O}(\sqrt n)$ queries for star-free regular languages extends to a variety of other settings given that the star-free languages enjoy a myriad of equivalent characterizations.  The characterization of star-free languages as first-order sentences over the natural numbers with the less-than relation shows that the algorithm for star-free languages is a broad generalization of Grover's algorithm.  See Section~\ref{sec:upperbounds} for the description and proof of the star-free algorithm and Section~\ref{sec:applications} and Appendix~\ref{sec:more_applications} for applications.
\item \textbf{Star-free algorithm from faster unstructured search:}  The $\tilde{O}(\sqrt n)$ algorithm for star-free languages results from many nested calls to Grover search, using the speedup due to multiple marked items.  However, a careful analysis reveals that whenever this speedup is required, the marked items are \emph{consecutive}.  We show that these Grover search calls can then be replaced by any unstructured search algorithm.  Therefore, any model of computation that has faster-than-brute-force unstructured search will have an associated speedup for star-free languages.  Consider, for example, the model of quantum computation of Aaronson, Bouland, Fitzsimons, and Lee in which non-collapsing measurements are allowed \cite{aaronson_bouland:2016}.  It was shown that unstructured search in that model requires at most $\tilde{O}(n^{1/3})$ queries, and therefore, star-free languages can be solved in $\tilde{O}(n^{1/3})$ queries as well.
\end{enumerate}

Finally, we stress that this trichotomy is only possible due to the extreme uniformity in the structure of regular languages.  In particular, the trichotomy does not extend to another basic model of computation, the context-free languages.  

\begin{theorem}
\label{thm:cfl_computable}
For all limit computable\footnote{We say that a number $c \in \mathbb{R}$ is \emph{limit computable} if there exists a Turing machine which on input $n$ outputs some rational number $T(n)$ such that $\lim_{n \to \infty} T(n) = c$.} $c \in [1/2,1]$, there exists a context-free language $L$ such that $Q(L) = O(n^{c + \epsilon})$ and $Q(L) = \Omega(n^{c - \epsilon})$ for all $\epsilon > 0$.  Furthermore, if an additive $\epsilon$-approximation to $c$ is computable in $2^{O(1/\epsilon)}$ time, then $Q(L) = \Theta(n^{c})$. In particular, any algebraic $c \in [1/2, 1]$ has this property. 
\end{theorem}
In fact, the converse also holds.
\begin{theorem}
\label{thm:cfl_converse}
Let $L$ be a context-free language such that $\lim_{n \to \infty} \frac{\log Q(L)}{\log n} = c$.  Then, $c$ is limit computable.
\end{theorem}
%TODO - It may be surprising that computability arises? 

\subsection{Proof Techniques}

Most of the lower bounds are derived from a dichotomy theorem for sensitivity---the sensitivity of a regular language is either $O(1)$ or $\Omega(n)$.  In particular, we show that the language of sensitive bits for a regular language is itself regular.  Therefore, by the pumping lemma for regular languages, we are able to boost any nonconstant number of sensitive bits to $\Omega(n)$ sensitive bits, from which the dichotomy follows.

The majority of the work required for the classification centers around the $\tilde{O}(\sqrt n)$ quantum query algorithm for star-free languages.  The proof is based on Sch\"utzenberger's characterization of star-free languages as those languages recognized by finite aperiodic monoids.  Starting from an aperiodic monoid, Sch\"utzenberger constructs a star-free language recursively based on the ``rank'' of the monoid elements involved.  Roughly speaking, this process culminates in a decomposition of any star-free language into star-free languages of smaller rank.  Although this decomposition does not immediately give rise to an algorithm, the notion of rank proves to be a particularly useful algebraic invariant.  Specifically, we use it to show that given a $\tilde{O}(\sqrt n)$-query algorithm for membership in some star-free language $L$, we can construct a $\tilde{O}(\sqrt n)$-query algorithm for $\Sigma^* L \Sigma^*$.  This ``infix'' algorithm is the key subroutine for much of the general star-free algorithm.

\subsection{Applications}
\label{sec:applications}

We give quantum quadratic speedups for several problems simply by showing that the underlying language is star free.  Consider the language $2\Sigma^{*}2 \backslash \Sigma^{*} 2 0^{*} 2 \Sigma^{*}$, where $\Sigma = \{ 0, 1, 2 \}$. We call this the \emph{dynamic AND-OR language}, for reasons which may not be evident from the regular expression alone. Think of the $2$'s as delimiting the string into some number of blocks over $\{ 0, 1 \}$. We take the OR of each block and the AND of those results to decide if the string is in the language. That is, if there is some pair of consecutive $2$'s with no intervening $1$, then that block evaluates to $0$, and the whole string is not in the language. It has long been known that the quantum query complexity of the AND-OR tree, or more generally Boolean formulas with constant depth, is $\Theta(\sqrt{n})$ \cite{hoyer:2003}. In that case, however, the tree or formula is fixed in advance and not allowed to change with the input.  Nevertheless, our quantum algorithm for star-free languages implies that even the dynamic version of the AND-OR language (as well as the dynamic generalization of constant-depth Boolean formulas \cite{barrington:1988}) can be decided with $\tilde{\Theta}(\sqrt{n})$ queries and, moreover, there is an efficient quantum algorithm. 

Next consider the language of balanced parentheses, where the parentheses are only allowed to nest $k$ levels deep. When $k$ is unbounded, this is called the Dyck language.  When $k = 1$ this is the language of strings of the form $\texttt{()()}\ldots \texttt{()}$, which has a simple Grover search speedup---search for $\texttt{((}$ or $\texttt{))}$. However, the language quickly becomes more interesting as $k$ increases. Nevertheless, for any constant $k$, this language is known to be star free \cite{crespi:1978}, and therefore has an $\tilde{O}(\sqrt n)$ quantum algorithm by our classification.  

Finally, we mention a few more examples of star-free languages (proofs in Appendix~\ref{sec:more_applications}).
\begin{itemize}[itemsep = -1pt]
\item[] \textit{Addition: } Given three binary numbers $x_0 y_0 z_0 x_1 y_1 z_1\ldots x_n y_n z_n$ as input, decide if $x + y = z$.
\item[] \textit{Word Break: } Given finite dictionary $D \subseteq \Sigma \cup \Sigma^2$, decide if word $x \in \Sigma^*$ is in $D^*$.
\item[] \textit{Grid Path: } Given a constant-height grid of cells, some of which are impassable, decide whether there is a path from the bottom left corner to the top right corner.
\end{itemize}

To the authors' knowledge, no quantum quadratic speedups for any of the previous problems were known prior to this publication.

\subsection{Related Work}

We are not the first to study regular languages in a query-complexity setting. One such example is work in property testing by Alon, Krivelevich, Newman, and Szegedy. They show that regular languages can be tested\footnote{We say a language $L$ is testable with constantly many queries if there exists a randomized algorithm such that given a word $w \in \Sigma^n$, the algorithm accepts $w$ if $w \in L$, and the algorithm rejects $w$ if at least $\epsilon n$ many positions of $w$ must be changed in order to create a word in $L$.  The algorithm is given $\tilde{O}(1/\epsilon)$ many queries to $w$.} with $\tilde{O}(1/\epsilon)$ queries \cite{alon:2001}. Interestingly, Alon et al.\ also show that there exist context-free grammars which do not admit constant query property testers \cite{alon:2001}. In Section~\ref{sec:cfl}, we show that context-free languages can have query complexity outside the trichotomy. 

A second example comes from work of Tesson and Th\'erien on the communication complexity of regular languages \cite{tesson:2003}. As with query complexity, several important functions in communication complexity happen to be regular, e.g., inner product, disjointness, greater-than, and index. They show that for several measures of communication complexity, the complexity is $\Theta(1)$, $\Theta(\log \log n)$, $\Theta(\log n)$, or $\Theta(n)$. Clearly, there are many parallels with this work, but surprisingly the classes of regular languages involved are different. Also, communication complexity is traditionally more difficult than query complexity, yet the authors appear to have skipped over query complexity---we assume because quantum query complexity is necessary to get an interesting result.

There are also striking parallels in work of Childs and Kothari, who conjecture a dichotomy for the quantum query complexity of minor-closed graph properties \cite{childs:2012}. Minor-closed graph properties are not, to our knowledge, directly related to regular languages, but they are morally similar in that both are very uniform---(almost) every part of the input is treated the same by the property. Childs and Kothari show that such properties have query complexity $\Theta(n^{3/2})$, except for forbidden subgraph properties which are $o(n^{3/2})$ and $\Omega(n)$, and are conjectured to be $\Theta(n)$. Even some of the proof techniques are similar---the proof that forbidden subgraph properties are $\Omega(n)$ could be phrased in terms of block sensitivity, like our $\Omega(\sqrt{n})$ lower bound for non-trivial languages. 

Finally, we are aware of one more result on the complexity of star-free languages prior to our work. It is possible to show that star-free languages have $o(n)$ quantum query complexity, just barely enough to separate them from non-star-free languages. This result is a combination of two existing results: Chandra, Fortune, and Lipton \cite{chandra:1983} show that star-free languages have (very slightly) super-linear size $\AC^{0}$ circuits; Bun, Kothari, and Thaler show that linear size $\AC^{0}$ circuits have (moderately) sublinear quantum query complexity \cite{bun:2018}. This connection was pointed out to us by Robin Kothari.

%!TEX root = ../reg-query.tex

\section{Background}
\label{sec:background}

This section introduces both regular languages and basic query complexity measures and their relationships.  In particular, we will focus on algebraic definitions of regular languages as they serve as the basis for many of the results in this paper.  Readers familiar with query complexity can skip much of the introduction on that topic, but may still want to read Section~\ref{sec:alphabet_size} on extending the complexity measures to larger alphabets.

\subsection{Regular languages}

The \emph{regular languages} are those languages that can be constructed from $\varnothing$, $\{ \varepsilon \}$, and singletons $\{ a \}$ for all $a \in \Sigma$ using the operations of concatenation (e.g., $AB$), union (e.g., $A \cup B$), and Kleene star\footnote{Let $A$ be a set of strings.  Define $A^* = \{ a_1 \ldots a_k : k \ge 0, a_i \in A\}$, that is, the concatenation of zero or more strings in $A$.  We will also use $A^+ = \{ a_1 \ldots a_k : k \ge 1, a_i \in A\}$ to capture one or more strings.} ($A^{*}$).  A \emph{regular expression} for a regular language is an explicit expression for how to construct the language, traditionally writing $|$ for alternation (instead of union), and omitting some brackets by writing $a$ for $\{ a \}$ and $\varepsilon$ for $\{ \varepsilon \}$.   For example, over the alphabet $\Sigma = \{0,1\}$, the OR function can be written as regular expression $\Sigma^* 1 \Sigma^*$, and the languages of all strings such that there are no two consecutive 1's is  $(0 | 10)^* (\varepsilon | 1)$.

The class of regular languages has extremely robust definitions and many equivalent characterizations.  For instance, some machine-based definitions\footnote{We assume familiarity with the basic machine models for regular languages---see \cite{sipser:2006} for an introduction.} include those languages accepted by deterministic finite automata (DFA), or by non-deterministic finite automata (NFA), or even by alternating finite automata. Regular languages also arise by weakening Turing machines, for example by making the machine read-only or limiting the machine to $o(\log \log n)$ space.

For our purposes, some of the most useful definitions of regular languages are algebraic in nature.  In particular, regular languages arise as the preimage of a subset of a finite monoid under monoid homomorphism.\footnote{A \emph{monoid} $(M, \cdot, 1_M)$ is a set $M$ closed under an associative binary operation $\cdot \colon M \times M \to M$ with an identity element $1_M \in M$.  A \emph{monoid homomorphism} is a map from one monoid to another that preserves multiplication and identity.}  First, we say that language $L \subseteq \Sigma^{*}$ is \emph{recognized} by a monoid $M$ if there exists a monoid homomorphism $\varphi \colon \Sigma^{*} \to M$ (where $\Sigma^{*}$ is a monoid under concatenation) and a subset $S \subseteq M$ such that 
$$
L = \{ w \in \Sigma^{*} : \varphi(w) \in S \} = \varphi^{-1}(S).
$$
\begin{theorem}[folklore]
\label{thm:reg_lang_finite_monoid}
A language is recognized by a finite monoid iff it is regular. 
\end{theorem}

In fact, starting from a regular language, we can specify a finite monoid recognizing it through the so-called syntactic congruence.  Given language $L \subseteq \Sigma^{*}$, the \emph{syntactic congruence} is an equivalence relation $\sim_{L}$ on $\Sigma^{*}$ such that $x \sim_L y$ if 
$$
\forall u, v \in \Sigma^{*}, \ uxv \in L \iff uyv \in L.
$$ 
Thus, $\sim_L$ divides $\Sigma^{*}$ into equivalence classes.  Furthermore, $\sim_L$ is a monoid congruence because $u \sim_L v$ and $x \sim_L y$ imply $ux \sim_L vy$. This means the equivalence classes of $\Sigma^{*}$ under $\sim_L$ are actually \emph{congruence classes} (because they can be multiplied), defining a monoid $M_{L}$ which we call the \emph{syntactic monoid of $L$}. Finally, it is not hard to see that the map $\varphi \colon \Sigma^{*} \to M_{L}$, from a string to its congruence class, is a homomorphism.   Therefore, by Theorem~\ref{thm:reg_lang_finite_monoid}, the syntactic monoid for any regular language is finite.

The most important subclass of regular languages are the \emph{star-free languages}.   These languages are recognized by a variant of regular expressions where complement ($\overline{A}$) is allowed but Kleene star is not.  We call these \emph{star-free regular expressions}.  For convenience, star-free regular expressions sometimes contain the intersection operation since it follows by De Morgan's laws.

Note that star-free languages are not necessarily finite. For example, $\Sigma^{*}$ can also be expressed as $\overline{\varnothing}$, the complement of the empty language. Similarly, $0^{*}$ is $\overline{\overline{\varnothing} (\Sigma \backslash \{ 0 \}) \overline{\varnothing}}$, the set of strings which do not contain a string other than $0$.  Once again, an algebraic characterization of star-free languages will be particularly useful for us.  First, we say that a monoid $M$ is \emph{aperiodic} if for all $x \in M$ there exists an integer $n \geq 0$ such that $x^{n} = x^{n+1}$.

\begin{theorem}[Sch\"{u}tzenberger~\cite{schutzenberger:1965}]
\label{thm:schutzenbergers_theorem}
A language is recognized by a finite aperiodic monoid iff it is star free. 
\end{theorem}

We also define a subset of the star-free languages, which we call the \emph{trivial languages}.  Intuitively, the trivial languages are those languages for which membership can be decided by the first and last characters of the input string,\footnote{More generally, trivial languages are decided by a constant size prefix and/or suffix of the input, but the processing we do to formalize the trichotomy theorem compresses those substrings to length 1. See Section~\ref{sec:formalstatement}.} which we formalize as those languages accepted by \emph{trivial regular expressions}.  A trivial regular expression is any Boolean combination of the languages $a | a \Sigma^* a$, $a \Sigma^*b$, and $\varepsilon$ for $a \neq b \in \Sigma$.  

The algebraic characterization of trivial languages will need to use both the properties of the monoid \emph{and} the properties of the homomorphism onto the monoid.  To that end, we say that language $L \subseteq \Sigma^{*}$ is \emph{recognized} by a monoid homomorphism $\varphi \colon \Sigma^{*} \to M$ if  $L = \{ w \in \Sigma^{*} : \varphi(w) \in S \} = \varphi^{-1}(S)$ for some subset $S \subseteq M$.  Finally, a monoid $M$ is a \emph{rectangular band} if for $r,s,t \in M$, each element is idempotent, $r^2 = r$, and satisfies the rectangular property, $rst = rt$.

\begin{theorem}[Appendix~\ref{sec:regex_parallels}]
\label{thm:trivial_lang_rectangular_band}
A language is recognized by morphism $\varphi$ such that $\varphi(\Sigma^+)$ is a finite rectangular band iff it is trivial. 
\end{theorem}

\subsection{Query complexity}
This section serves as a brief overview of query complexity, a model of computation where algorithms are charged based on the number of input bits they reveal (the input is initially hidden) rather than the actual computation being done. To model that the input is hidden, all query algorithms must access their inputs via an indexing oracle---a function which takes some index and outputs the value of the corresponding input bit.  We use the standard notion of oracles in the quantum setting.  That is, for oracle function $\mathcal{O} \colon \{0,1\}^n \to \{0,1\}$, the quantum algorithm can apply the $(n+1)$-qubit transformation which flips the last qubit if $\mathcal{O}$ applied to the first $n$ qubits evaluates to 1.  

Formally, the \emph{quantum query complexity} of a function $f \colon \Sigma^* \to \{0,1\}$ is a function $Q(f) \colon \mathbb N \to \mathbb N$ such that $Q(f)(n)$ is the minimum number of oracle calls for a quantum circuit to decide (with bounded error) the value of $f$ for input strings of length $n$.  An astute reader may notice that we only defined the indexing function over bits and that regular languages are defined over arbitrary finite alphabets $\Sigma$. However, one can always transform the function so that each symbol of $\Sigma$ is encoded by $\lceil \log_2 |\Sigma| \rceil$ bits.  In fact, we will show later that this only affects the query complexity by a constant factor for regular languages.

One can similarly define \emph{deterministic query complexity} ($D$), \emph{bounded-error randomized query complexity} ($R$), and \emph{zero-error randomized query complexity} ($R_0$) by counting the number of input symbols accessed in these models.  Closely related to quantum query complexity is a notion of approximation by polynomials called approximate degree, denoted $\approxdeg(f)$.  The \emph{approximate degree} of a function $f \colon [k]^{n} \to \{ 0, 1 \}$ is the minimum degree of a polynomial $p(x_1, \ldots, x_n)$ such that $|p(x_1, \ldots, x_n) - f(x_1, \ldots, x_n)| \leq \frac{1}{3}$ for all $x_1, \ldots, x_n \in [k]$.

We conclude by defining several query complexity measures which are useful tools in proving lower bounds in the more standard models of computation above.  Fix a function $f \colon \Sigma^{*} \to \{ 0, 1 \}$.  Let $x \in \Sigma^n$ be some input.  We say that some input symbol $x_i$ is \emph{sensitive} if changing only $x_i$ changes the value of the function on that input.  The \emph{sensitivity} of $x$ is equal to its total number of sensitive symbols. The \emph{sensitivity of $f$}, denoted $s(f)$, is the maximum sensitivity over all inputs $x$. 

Similarly, the \emph{block sensitivity} at an input is the maximum number of disjoint blocks (i.e., subsets of the input bits) such that changing one entire block changes the value of the function. The \emph{block sensitivity of $f$}, denoted $bs(f)$, is the maximum block sensitivity over all inputs $x$.

A \emph{certificate} is a partial assignment of the input symbols such that $f$ evaluates to the same value on all inputs consistent with the certificate. The certificate complexity of an input is the minimum certificate size (i.e., the number of bits assigned in the partial assignment). The \emph{certificate complexity of $f$}, denoted $C(f)$, is the maximum certificate complexity over all inputs. 

Finally, when clear from context, we will often let a language denote its characteristic function when used as an argument in the various complexity measures.  For example, for language $L \subseteq \Sigma^*$, we will write $Q(L)$ as the quantum query complexity of the function $f_L \colon \Sigma^* \to \{0,1\}$ where $f(x) = 1$ iff $x \in L$.

\subsubsection{Relationships}

There are many relationships between the different complexity measures that will be useful throughout this paper.  For example, the proposition below follows from the fact that some models of computation can easily simulate others.
\begin{proposition}[\cite{beals:2001}]
\label{prop:qcinclusions}
For all $f \colon \{0,1\}^* \to \{0,1\}$, 
$$
\frac{1}{2} \approxdeg(f) \leq Q(f) \leq R(f) \leq R_0(f) \leq D(f).
$$
\end{proposition}

In Section~\ref{sec:dichotomies}, we prove a dichotomy theorem for block sensitivity---it is either $O(1)$ or $\Omega(n)$.  This is particularly useful since nearly all complexity measures are polynomially related to block sensitivity:\begin{theorem}[\cite{buhrman:2002}]
\label{thm:bs_relations}
For all $f \colon \{0,1\}^* \to \{0,1\}$, we have the following relationships for block sensitivity:
\begin{center}
\setlength\tabcolsep{2pt}
\begin{tabular}{r c l l | l r c l}
\multicolumn{3}{c}{Lower bounds} & \hspace{.1em} &  \hspace{1em} & \multicolumn{3}{c}{Upper bounds} \\ \hline
$C(f)$ & $\geq$ & $\bs(f)$ & & &  $s(f)$ & $\leq$ & $bs(f)$ \\
$\approxdeg(f)$ & $=$ & $\Omega(\sqrt{bs(f)})$ & & &  $C(f)$ & $\leq$ & $bs(f)^2$ \\
$R(f)$ & $=$ & $\Omega(\bs(f))$ & & & $D(f)$ & $\leq$ & $bs(f)^3$.
\end{tabular}
\end{center}

\end{theorem}

Notice that for nearly all complexity measures $M$, we have $bs(f)^a \leq M(f) \leq bs(f)^b$ for some constants $a, b \geq 0$. The exception is sensitivity, for which it is famously open whether a polynomial in sensitivity upper bounds block sensitivity. There is, however, an \emph{exponential} relation due to Simon. 
\begin{theorem}[Simon \cite{simon:1983}]
\label{thm:simonsensitivity}
For all $f \colon \{0,1\}^* \to \{0,1\}$,  $bs(f) = O(s(f) 4^{s(f)})$.
\end{theorem}

\begin{corollary}
\label{cor:qcallconstant}
If any query complexity measure in $\{ s, bs, C, D, R_0, R, Q, \approxdeg \}$ is $O(1)$, then all of them are $O(1)$.
\end{corollary}

\subsubsection{Alphabet size}
\label{sec:alphabet_size}

In this section, we discuss how alphabet size affects the various query measures.  Recall that the query complexity measures above are usually defined for \emph{Boolean} functions. Nevertheless, we would like to extend the known relationships between the complexity measures to functions over larger (yet constant) alphabets.  While it is true that many of these relationships generalize without too much work, we would like to avoid reproving the results one at a time.

Our solution is to simply encode symbols of $\Sigma$ as binary strings of length $\lambda := \lceil \log |\Sigma| \rceil$.  If the size of the alphabet $\Sigma$ is not a power of two, we can simply map the extra binary strings to arbitrary elements of $\Sigma$.  This maps a language $L \subseteq \Sigma^*$ to a language $\binL \subseteq \{ 0, 1 \}^{*}$ over binary strings.  Since regular languages are closed under inverse morphism, $\binL$ is regular if $L$ is regular.

It is also easy to see that almost all complexity measures are changed by at most a constant factor when converting to a binary alphabet.  For example, $D(L)(n) \leq D(\binL)(\lambda n)$ since for any bit we look at, there is some symbol we can examine that tells us that bit. In the other direction, $D(\binL)(n) \leq \lambda D(L)(\lambda n)$, since we can query the entire encoding of any symbol we query. Similarly, the encoding changes $R_0$, $R$, $Q$, $s$, $C$, and (with some additional work) $\approxdeg$, by at most a constant factor. The exception is block sensitivity. 

It is clear that $bs(L)(n) \leq bs(\binL)(\lambda n)$, since for any sensitive block of symbols there is some way to flip it, and this changes some block of bits. In the other direction, a block of sensitive bits gives a block of sensitive symbols in the obvious way, but then disjoint blocks of bits will not necessarily map to disjoint blocks of symbols, so it is difficult to say more for general languages. 

\begin{theorem}
\label{thm:alphabetsize}
Let $L \subseteq \Sigma^{*}$ be a regular language. Then, there exists constant $c$ such that $bs(L)(n) \ge c \cdot bs(\binL)(\lambda n)$ for all $n$. 
\end{theorem}
\begin{proof}
We borrow a dichotomy result\footnote{Note that Corollary~\ref{cor:sensitivity} is true for any alphabet size and does not depend on Theorem~\ref{thm:alphabetsize}, so the argument is not circular.} from Section~\ref{sec:dichotomies}, namely Corollary~\ref{cor:sensitivity}---any flat regular language has sensitivity either $O(1)$ or $\Omega(n)$.  Since $L$ is a regular language and not necessarily flat, we also borrow Theorem~\ref{thm:flattening} from Section~\ref{sec:formalstatement}---membership in $L$ reduces to membership in some flat language based on some finite suffix of the input string.  Therefore, for every length $n$, the sensitivity $s(L)$ is either constant or $\Omega(n)$, which we use to split the proof into two cases.

If the sensitivity $s(L)$ is constant, then $s(\binL)$ is also constant.  This implies that $bs(\binL)$ is constant by Theorem~\ref{thm:simonsensitivity}.   Therefore, $bs(L)$ is also constant since $bs(L)(n) \le bs(\binL)(\lambda n)$.  
If the sensitivity $s(L)$ is not constant, then it is linear by the dichotomy theorem.  Therefore, $s(L)(n) \leq bs(L)(n) \leq bs(\binL)(\lambda n)$ implies block sensitivity is linear for both languages from which the theorem follows. 
\end{proof}
With this theorem, every regular language and its encoding have the same complexity for all of the measures we are interested in, up to constants.  Therefore, we will lift known relationships between complexity measures in the Boolean setting to the general alphabet setting without further comment.

%!TEX root = ../reg-query.tex

\section{Formal Statement}
\label{sec:formalstatement}

The na\"{i}ve version of the trichotomy theorem states that the quantum query complexity of a regular language is always $\Theta(1)$, $\tilde{\Theta}(\sqrt{n})$, or $\Theta(n)$. Unfortunately, this is not strictly true. We now explain the difficulty and a technique which we call ``flattening'' that allows us to formalize this statement. 

Let us see why flattening is necessary.  Consider any language which has large quantum query complexity (e.g., parity) and take its intersection with $(\Sigma^2)^{*}$, the language of even length strings. When the input length is odd, we know without any queries that the string cannot be in the language. When the input length is even, we have to solve the parity problem, which requires $\Omega(n)$ queries. Thus, the query complexity oscillates drastically between $0$ and $\Theta(n)$ depending on the length of the input.  Strictly speaking, this means the complexity is neither $\Theta(1)$, $\tilde{\Theta}(\sqrt{n})$, nor $\Theta(n)$; the na\"{i}ve statement of the trichotomy is false. 

%Unfortunately, query complexity is non-uniform in the sense that the query algorithm may be different for each input length. For example, it could oscillate drastically between $0$ and $\Theta(n)$ depending on the parity of the length. 

We want to state the trichotomy only for languages which are length-independent. Fortunately, a DFA cannot count how many symbols it reads. With finite state, the best a DFA can do is count \emph{modulo} some constant. Thus, if there is any dependence on length, it is periodic.  Similarly, a language may have periodic dependence on position. For example, consider the language of all strings with exactly two $1$s. This language is star free and therefore has an $\tilde{O}(\sqrt n)$ quantum query algorithm. If we further require the $1$s to be an even distance apart, the language is no longer star free, but clearly has an $\tilde{O}(\sqrt{n})$ quantum query algorithm.  Flattening will reduce this language to a collection of star-free languages, and in general it will remove periodicities not inherent to the query complexity of the language.

Before continuing with flattening, we address a different way to handle \emph{length} dependence. That is, redefine the quantum query complexity of a function to be the minimum number of quantum oracle calls needed to compute the function on inputs of length \emph{up to} $n$ (rather than exactly $n$).  For this definition, notice that the quantum query complexity is nondecreasing. In Appendix~\ref{sec:monotone_query_complexity} we show that trichotomy theorem holds for \emph{all} regular languages under this definition as a simple consequence of Theorem~\ref{thm:classification}, the trichotomy theorem for \emph{flat} languages.  To be clear, we will continue to use the standard definition of quantum query complexity for the remainder of the paper.

\subsection{Flattening}

The main idea behind flattening is to eliminate a language's periodicities by dividing the strings into blocks.  For any string $x \in \Sigma^*$ of length $kn$, we can reimagine $x$ as a length-$n$ string over $\Sigma^k$.  This operation can be applied to a language by keeping only strings of length divisible by $k$ and projecting them to the alphabet $\Sigma^k$.  Flattening a regular language applies this operation to the language for some carefully chosen $k$ to be determined later.  Nevertheless, we argue that the language and its flattened version are essentially the same since we are simply blocking characters together.   We formalize this in the following theorem.

\begin{theorem}
\label{thm:flattening}
Let $L \subseteq \Sigma^{*}$ be a regular language recognized by a monoid $M$. There exists an integer $p \geq 2$ and a finite family of \emph{flat} regular languages $\{ L_i \}_{i \in I}$ over alphabet $\Sigma^{p}$ such that testing membership in $L$ reduces (in fewer than $p$ queries) to testing membership in $L_i$ for some $i$. Furthermore, the same monoid $M$ recognizes $L$ and every $L_i$ (although there may be a simpler monoid which recognizes $L_i$).
\end{theorem}
The full proof is in Appendix~\ref{sec:flattening} with the rest of the details about flattening a language. The key property of a flattened language is the following:

 \begin{property}
\label{property:flat}
Let $L \subseteq \Sigma^{*}$ be a flat regular language. For any non-empty string $x \in \Sigma^{+}$, and any non-zero length $k > 0$, there exists a string $y \in \Sigma^{k}$ of length $k$ such that for any $u, v \in \Sigma^{*}$, 
$$
uxv \in L \iff uyv \in L.
$$
That is, $x$ and $y$ belong to the same congruence class.
\end{property}

In other words, for any non-empty string $x$, we can replace (substring) occurrences of $x$ with some string of every (non-zero) length, without changing membership in the language. Notice that a flat regular language cannot have a length dependence, otherwise we would replace the first few letters with something slightly longer or shorter to reduce the problem to whichever nearby length is easiest. 

%We claim that for any language $L$, there is a finite period $p$ such that each $L_i$ is a \emph{flat} regular language. 

To summarize, any regular language can be reduced (or \emph{flattened}) to a collection of flat regular languages. Some of these languages may be easier than others, but they are all length-independent, and thus suitable for our trichotomy theorem. See Appendix~\ref{sec:flattening} for details.

\subsection{Formal Statement of Main Result}

We are now ready to formally state Theorem~\ref{thm:classification}.  Technically, there are a few regular languages (even flat languages), which can be decided with zero queries, strictly from the length of the input. This divides the languages into the following \emph{four} classes (i.e., a \emph{tetrachotomy}).

\begin{reptheorem}{thm:classification}
Every flattened regular language has quantum query complexity $0$, $\Theta(1)$, $\tilde{\Theta}(\sqrt{n})$, or $\Theta(n)$ according to the smallest class in the following hierarchy that contains the language.
\begin{itemize} [itemsep = 0pt]
\item Degenerate: One of the four languages $\varnothing$, $\varepsilon$, $\Sigma^{*}$, or $\Sigma^+$.
\item Trivial: The set of languages which have trivial regular expressions. 
\item Star free: The set of languages which have star-free regular expressions.
\item Regular: The set of languages which have regular expressions. 
\end{itemize}
Note that each class is contained in the next.  Furthermore, the quantum time complexity of each class matches its query complexity.
\end{reptheorem}

%
%\begin{reptheorem}{thm:classification}
%Every flattened regular language is either 
%\begin{itemize} [itemsep = 0pt]
%\item Degenerate: One of the four languages $\varnothing$, $\varepsilon$, $\Sigma^{*}$, or $\Sigma^+$.
%\item Trivial: The set of languages which have trivial regular expressions. 
%\item Star free: The set of languages which have star-free regular expressions.
%\item Regular: The set of languages which have regular expressions. 
%\end{itemize}
%
%The regular languages contain the star-free languages, which contain the trivial languages, which contain the degenerate languages. Furthermore, the quantum query complexity of a language in each class---provided that it is contained in no smaller class---is $0$, $\Theta(1)$, $\tilde{\Theta}(\sqrt n)$, or $\Theta(n)$, respectively.
%\end{reptheorem}

Nevertheless, we refer to this classification as a \emph{trichotomy}. We either think of degenerate and trivial languages under the category of ``constant query regular languages" or, alternatively, disregard the degenerate languages entirely because they are uninteresting.

As it turns out, the regular expression descriptions, some of which were already mentioned in Section~\ref{sec:background}, are not particularly useful for the classification. We will prefer the following algebraic/monoid definitions of the languages, and use them throughout. We prove they coincide with the regular expression characterizations in Appendix~\ref{sec:regex_parallels}.

\begin{theorem}
\label{thm:monoid_characterizations}
Let $L$ be a regular language.
\begin{itemize} [itemsep = 0pt]
\item $L$ is degenerate iff it is recognized by morphism $\varphi$ such that $|\varphi(\Sigma^+)| = 1$. 
\item $L$ is trivial iff it is recognized by morphism $\varphi$ such that $\varphi(\Sigma^+)$ is a finite rectangular band. 
\item $L$ is star free iff it is recognized by a finite aperiodic monoid. 
\item $L$ is regular iff it is recognized by a finite monoid. 
\end{itemize}
\end{theorem}

\subsection{Structure of the proof}

We separate the proof of the trichotomy into two natural pieces: upper bounds (Section~\ref{sec:upperbounds}) and lower bounds (Section~\ref{sec:lowerbounds}).   The upper bounds are derived directly from the monoid characterizations of the various classes.  Given a flat language, we construct explicit algorithms using at most 0 queries for degenerate languages, 2 queries for trivial languages, $\tilde{O}(\sqrt n)$ queries  for star-free languages, and $n$ queries for regular languages.

The lower bound section aims to prove that these are the \emph{only} possible classes.  First, we show that any non-degenerate language requires at least one quantum query.  We then show that any nontrivial language requires $\omega(1)$ quantum queries.  At this point, we will appeal to a dichotomy theorem for the block sensitivity of regular languages, which we prove in Section~\ref{sec:dichotomies}.  From this dichotomy and standard relationships between the complexity measures, we get that any regular language requiring $\omega(1)$ quantum queries actually requires $\Omega(\sqrt n)$ queries.  Finally, we show that any non-star-free language requires $\Omega(n)$ queries, completing the proof.

%!TEX root = ../reg-query.tex

\section{Upper Bounds}
\label{sec:upperbounds}

In this section, we will describe the algorithms for achieving the query upper bounds in Theorem~\ref{thm:classification}.  As a warm-up, we will first consider every class besides the star-free languages.  Each algorithm will follow trivially from the monoid characterization of each class. 

\begin{proposition}
\label{prop:easy_upper_bounds}
Any regular language has an $O(n)$-time deterministic algorithm.  The trivial languages have constant-time deterministic algorithms.  The degenerate languages have $0$-query deterministic algorithms.\footnote{Note, the power of constant-time algorithms depends on the particular model of computation.  We assume a RAM model where the length of the input string is given, and arithmetic on indices can be performed in constant time.}
\end{proposition}
\begin{proof}
Let $L \subseteq \Sigma^*$ be a regular language.  Let $\varphi$ be the homomorphism onto its syntactic monoid $M_L$ such that $L = \{ \varphi^{-1}(s) : s \in S \subseteq M_L \}$.  Let $x = x_1 \ldots x_n \in \Sigma^n$.  We have that $x \in L$ iff $\varphi(x_1) \varphi(x_2) \ldots \varphi(x_n) \in S$.  Since $M_L$ is finite and $\varphi$ is specified by a finite mapping from characters to monoid elements, this product is computable in linear time.

Suppose $L$ is trivial.  Consider input $x = a y b$ where $a, b \in \Sigma$ and $y \in \Sigma^*$.  By the rectangular band property, we have $\varphi(x) = \varphi(a) \varphi(y) \varphi(b) = \varphi(a) \varphi(b)$.  That is, $x \in L$ iff $\varphi(ab) \in S$.

Suppose $L$ is degenerate.  Consider some input $x \in \Sigma^*$.  If $|x| = 0$, then $x \in L$ iff $\varphi(\varepsilon) \in S$.  If $|x| > 0$, then $\varphi(x) \in \varphi(\Sigma^+) = \{s\}$ so $x \in L$ iff $s \in S$.  Since the query algorithm knows the length in advance, no queries are needed to determine the membership of $x$.
\end{proof}

Of course, the existence of these deterministic algorithms implies their corresponding query upper bounds as well.  Much more interesting is the $\tilde{O}(\sqrt{n})$ quantum algorithm for star-free languages to which the remainder of this section is dedicated.  Much like Proposition~\ref{prop:easy_upper_bounds}, we will use the monoid characterization as our starting point for the algorithm; however, before delving directly into the details of the algorithm, we give some techniques and ideas that will be pervasive throughout.

\subsection{Proof techniques}
\label{sec:proof_techniques}

In this section, we introduce a basic substring search operation and a decomposition theorem (due to Sch\"utzenberger) for aperiodic monoids.

\subsubsection{Splitting and infix search}

Consider the language $L = \Sigma^{*} 2 0^{*} 2 \Sigma^{*}$ over the alphabet $\Sigma = \{ 0, 1 , 2 \}$, that is, the problem of finding a substring of the form $2 0^{*} 2$. We call the problem of finding a contiguous substring satisfying a predicate \emph{infix search}. Since $L$ is star free, our trichotomy theorem implies that infix search for the language $2 0^* 2$ is possible with $\tilde{O}(\sqrt n)$ queries.

Consider the following algorithm for $L$: Grover search for an index $i$ in the middle of a substring $20^{*}2$, searching outwards to verify that there is a substring of the form $20^{*}$ immediately before the index (\emph{suffix search}) and a substring of the form $0^{*} 2$ immediately after (\emph{prefix search}). More precisely, we can use Grover search to check whether a substring is all $0$s, then binary search to determine how far the $0$s extend on either side of the index, and finally check for $2$s on either end.

We introduce a few ideas necessary to prove this algorithm for $L$ is efficient, and to generalize it to arbitrary languages. The first tool we need is Grover search, to help us search for the position of the substring. In particular, we use a version of Grover search which is faster when there are multiple marked items.\footnote{
	In this section, we will need the speedup from multiple marked items. However, whenever we require the speedup, the marked items will be consecutive. In this case, we can derive the same speedup from any $\tilde{O}(\sqrt{n})$ unstructured search algorithm by searching over indices at fixed intervals (a ``grid" on the input). In more detail: we search for a grid size $G$, starting from $n$ and halving until $G$ is less than the number of consecutive marked items (which is unknown). Hence, the set of indices divisible by $G$ will intersect some marked item and the search on $n/G$ indices will succeed in $\tilde{O}(\sqrt{n/G})$ queries.  Since the last search dominates the runtime, the entire procedure requires $\tilde{O}(\sqrt{n/t})$ queries. \\
	\indent In fact, there are other models of computation where unstructured search uses $\tilde{O}(n^c)$ queries for $c \neq 1/2$ (for instance, \cite{aaronson_bouland:2016}).  It will turn out that the procedure described above still accelerates search for multiple consecutive marked items.  This will translate to an $\tilde{O}(n^c)$-query algorithm for star-free languages.  In particular, the runtime in Theorem~\ref{thm:scott_trick} becomes $\tilde{O}(n^c)$.
} 
\begin{theorem}[Grover search]
\label{thm:grover}
Given oracle access to a string of length $n$ which is 1 on at least $t \ge 1$ indices, there exists a quantum algorithm which returns a random index on which the oracle evaluates to 1 in $O(\sqrt{n/t})$ queries with constant probability.
\end{theorem}

Next, the solution to $\Sigma^{*} 20^{*} 2 \Sigma^{*}$ used the fact that given an index, we can search outwards for a substring $20^{*}$ before the index and $0^{*}2$ after. Notice that the index has ``split" the regular language $20^{*}2$ into two closely related languages. It is not clear every language has this property, so we introduce a notion of splitting for arbitrary regular languages.

\begin{definition}
We say that a language $L \subseteq \Sigma^*$ \emph{splits} if there exists a constant $k$ and languages $A_1, \ldots, A_k, B_1, \ldots, B_k$ such that $L = \bigcup_{i=1}^{k} A_i B_i$ and for all $x \in L$ and decompositions $x = uv$, there exists $1 \leq i \leq k$ such that $u \in A_i$ and $v \in B_i$. We say $L$ \emph{splits as} $\bigcup_{i=1}^{k} A_i B_i$ to succinctly introduce the languages. 
\end{definition}

Formally, $2 0^* 2$ splits as $(20^{*}2) \varepsilon \cup (2 0^*)(0^* 2) \cup \varepsilon (20^{*}2)$. In fact, \emph{every} star-free language $L\subseteq \Sigma^{*}$ splits as $\bigcup_{i=1}^{k} A_i B_i$ where the $A_i$ and $B_i$ are also star free.  We will prove this in the next section in Theorem~\ref{thm:monoid_split}.  We delay the proof until we have the definitions to show that the languages $A_i$ and $B_i$ are in some sense no harder than the language $L$ itself.

%Of course, this notion is only helpful if we can always split a star free language. The following lemma guarantees that we can split.
%\begin{lemma}
%	\label{lem:split}
%	Every regular language $L\subseteq \Sigma^{*}$ splits as $\bigcup_{i=1}^{k} A_i B_i$. The $A_i$, $B_i$ are star free iff $L$ is star free.
%\end{lemma}

Supposing we can determine membership for $\Sigma^{*} A_i$ and $B_i \Sigma^{*}$ efficiently, a combination of Grover search and exponential search will solve the infix search problem, as shown below. 

\begin{theorem}[Infix search]
\label{thm:scott_trick}
Let language $L \subseteq \Sigma^{*}$ split as $\bigcup_{i=1}^{k} A_i B_i$.   Suppose $Q(\Sigma^{*}A_i)$ and $Q(B_i \Sigma^{*})$ are $\tilde{O}(\sqrt{n})$ for all $i \in \{1, \ldots, k\}$. Then, $Q(\Sigma^* L \Sigma^*) = \tilde{O}(\sqrt{n})$.
\end{theorem}
\begin{proof}
We perform an exponential search---doubling $\ell$ with $\ell$ initially set to 1---until the algorithm succeeds. Let $x$ be the input and suppose there is a substring of $x$ belonging to $L$ of length at least $\ell$ and at most $2\ell$, for some power of two $\ell$. Search for an index $j$ such that $x_{j - 2\ell} \cdots x_{j - 1} \in \Sigma^* A_i$ and $x_j \cdots x_{j + 2\ell-1} \in B_i \Sigma^*$ for some $i = 1, \ldots, k$. This implies the substring $x_{j-2\ell} \cdots x_{j+2\ell-1}$ is in $\Sigma^{*} A_i B_i \Sigma^{*} \subseteq \Sigma^{*} L \Sigma^{*}$.

Since testing each index requires at most $\tilde{O}(\sqrt \ell)$ queries and $k$ is constant, there are $\tilde{O}(\sqrt \ell)$ queries to the string to test a particular index $j$. Recall that we assumed the matching substring has length at least $\ell$, and thus, there are $\ell$ indices of $x$ for which the prefix/suffix queries will return true. Hence, there are at most $O(\sqrt{n / \ell})$ total Grover iterations (Theorem~\ref{thm:grover}), and the final algorithm requires only $\tilde{O}(\sqrt n)$ queries. 
\end{proof}

\subsubsection{Aperiodic monoids and Sch\"utzenberger's proof}

At its core, the algorithm for star-free languages uses one direction of Sch\"utzenberger's theorem for star-free languages, which we recall from Section~\ref{sec:background}.

\begin{reptheorem}{thm:schutzenbergers_theorem} %[Sch\"utzenberger \cite{schutzenberger:1965}]
If language $L$ is recognized by a finite aperiodic monoid, then $L$ is star free.
\end{reptheorem}

We will show that Sch\"utzenberger's proof can be modified to produce a $\tilde{O}(\sqrt n)$ algorithm for any star-free language starting from the aperiodic monoid recognizing it.  Central to this modification will be the notion of splitting introduced in the previous section.  In this section we give the basic prerequisites and outline for Sch\"utzenberger's proof which will eventually culminate in a formal justification of splitting based on the properties of aperiodic monoids.

Let $M$ be a finite aperiodic monoid recognizing some language $L \in \Sigma^*$.  Recall that 
$
L = \varphi^{-1}(S) = \bigcup_{m \in S} \varphi^{-1}(m),
$ where $\varphi \colon \Sigma^{*} \to M$ is a surjective monoid homomorphism, and $S \subseteq M$ is some subset of the monoid.  Thus, to show that $L$ is star free, it suffices to show that $\varphi^{-1}(m)$ is star free for each $m \in M$.  

One of the central ideas in Sch\"utzenberger's proof is to consider these languages in order of the size of the ideal\footnote{Let $M$ be a monoid and $I \subseteq M$ be a subset. We say $I$ is a \emph{right ideal} if $IM = I$, $I$ is a \emph{left ideal} if $MI = I$, and $I$ is an \emph{ideal} if $MIM = I$.  For example, for any $m \in M$, $mM$ is a right ideal, $Mm$ is a left ideal, and $MmM$ is an ideal.} they generate.  Formally, Sch\"utzenberger's proof is an induction on the \emph{rank} of $m$, defined as
$$
\rho(m) := |M - MmM|,
$$
that is, the number of elements not in $M m M = \{ amb : a \in M, b \in M \}$.  For example, $\rho(1) = 0$.  Rank is a particularly useful measure of progress in the induction due to the following proposition:

%By Theorem~\ref{thm:}, recall that any regular language $L \subseteq \Sigma^{*}$ can be expressed as $\varphi^{-1}(S)$ where $\varphi \colon \Sigma^{*} \to M$ is a surjective monoid homomorphism onto a finite monoid, and $S \subseteq M$ is some subset of the monoid. Therefore, to determine membership in $L$, it suffices to consider languages of the form $\varphi^{-1}(m)$ for each $m \in S$.
%
%Sch\"utzenberger's proves star-freeness of each language $\varphi^{-1}(m)$ by induction on the ``rank" of the monoid elements.  The rank of a monoid element is a function of the size of its ideal

\begin{proposition}
\label{prop:rproduct}
For any $p, q \in M$ we have $\rho(p), \rho(q) \leq \rho(pq)$. 
\end{proposition}
\begin{proof}
%In particular, we can take a string $x$ (e.g., a prefix, suffix, or substring of the input) and consider $\varphi(x) M$, $M \varphi(x)$, and $M \varphi(x) M$. Notice that as the string gets longer, the ideal monotonically decreases: 
%\begin{align*}
%\varphi(xa) M = \varphi(x) \varphi(a) M &\subseteq \varphi(x) M, \\
%M \varphi(ax) = M \varphi(a) \varphi(x) &\subseteq M \varphi(x),
%\end{align*}
%$$
%M \varphi(x) \varphi(a) M \subseteq M \varphi(x) M \supseteq M \varphi(a) \varphi(x) M.
%$$
$MpqM \subseteq MpM$, so $M - MpqM \supseteq M - MpM$. Therefore, $\rho(p) \leq \rho(pq)$. Similarly, $\rho(q) \leq \rho(pq)$. 
\end{proof}

It will turn out that only the identity of the monoid $M$ has rank 0. First, we show that a product of monoid elements is the identity if and only if every element is the identity. 
\begin{proposition}
	\label{prop:unique_rank_zero_element}
	For elements $p_1, \cdots, p_n \in M$ in an aperiodic monoid $M$, if $p_1 \cdots p_n = 1$ then $p_1 = \cdots = p_n = 1$. 
\end{proposition}
\begin{proof}
	It suffices to prove the result for $n = 2$ and induct. Suppose $1 = pq$, and then by repeated substitution, 
	$$
	1 = pq = p 1 q = p^2 q^2 = \cdots = p^{i} q^{i},
	$$
	for any $i$. Since the monoid is aperiodic, there exists $n \geq 0$ such that $p^{n+1} = p^{n}$. Therefore, 
	$$
	p = p (p^{n} q^{n}) = p^{n+1} q^{n} = p^{n} q^{n} = 1.
	$$
	By symmetry, $q$ is also the identity.
\end{proof}

\begin{corollary}
\label{cor:shitty_corollary}
Let $M$ be a finite aperiodic monoid.  For any $m \in M$, $\rho(m) = 0$ iff $m=1$.
\end{corollary}
\begin{proof}
Suppose that $\rho(m) = 0$ for some monoid element $m \in M$.  By the definition of rank, we have that $M = MmM$, and in particular $1 \in M$ implies $1 = amb$ for some $a, b \in M$. By Proposition~\ref{prop:unique_rank_zero_element}, $a = b = m = 1$. 
\end{proof}

%This also implies that $\varphi^{-1}(1)$ is trivial: $\varphi$ must map every character to $1$. That is, 
%$$
%\varphi^{-1}(1) = \{ a \in \Sigma : \varphi(a) = 1 \}^{*}. 
%$$
%This language is both trivially star free and has a $O(\sqrt{n})$ quantum query algorithm based on a straightforward application of Grover search. 

It is not hard to see that $\varphi^{-1}(1)$ is star free.  For $\rho(m) > 0$, Sch\"utzenberger decomposes $\varphi^{-1}(m)$ into a Boolean combination of star-free languages with strictly smaller rank, completing the proof.
%$r$, which can be seen to prove the conditions of Theorem~\ref{thm:witnessm}. 
To avoid recapitulating all of Sch\"utzenberger's proof, we simply quote the main decomposition theorem.
\begin{theorem}[Decomposition Theorem]
\label{thm:schutzdecomp}
For any $m \in M$, 
$$
\varphi^{-1}(m) = (U \Sigma^{*} \cap \Sigma^{*} V) \backslash (\Sigma^{*} C \Sigma^{*} \cup \Sigma^{*} W \Sigma^{*}).
$$
where
\begin{align*}
U &= \bigcup_{(r,a) \in E} \varphi^{-1}(r) a  \\
V &= \bigcup_{(a,r) \in F}  a \varphi^{-1}(r) \\
C &= \{ a \in \Sigma : m \notin M \varphi(a) M \} \\
W &= \bigcup_{(a,r,b) \in G}  a \varphi^{-1}(r) b 
\end{align*}
and
\begin{align*}
E &= \{ (r,a) \in M \times \Sigma : r \varphi(a) M = mM, rM \neq mM \}, \\
F &= \{ (a,r) \in \Sigma \times M : M \varphi(a) r = Mm, Mr \neq Mm \}, \\
G &= \{ (a,r,b) \in \Sigma \times M \times \Sigma : m \in (M \varphi(a) r M \cap M r \varphi(b) M) \backslash M \varphi(a) r \varphi(b) M \}.
\end{align*}
Furthermore, for all $r \in M$ appearing in $E$, $F$, or $G$, $\rho(r) < \rho(m)$. 
\end{theorem}

%Let us relate these languages back to Theorem~\ref{thm:witnessm}. The language $U \Sigma^{*}$ matches a string $x$ if there is some prefix whose right ideal is $mM$. Similarly, $\Sigma^{*} V$ matches strings such that some prefix has left ideal $Mm$. The remaining two languages, $\Sigma^{*} C \Sigma^{*}$ and $\Sigma^{*} W \Sigma^{*}$, match strings where some substring is in $J_m$. The first, $\Sigma^{*} C \Sigma^{*}$, is matches if the substring in $J_m$ is a single character, and $\Sigma^{*} W \Sigma^{*}$ is for matching substrings of length $2$ or more.
%
%At this point, Sch\"utzenberger is done; he has decomposed an arbitrary $\varphi^{-1}(m)$ into a star-free expression of simpler $\varphi^{-1}(n)$ which, by induction, are also star-free languages. This completes the induction and finishes the theorem.\footnote{Actually, there are two directions to prove in Sch\"utzenberger's theorem, and this only finishes the proof of one direction.} 

To see the decomposition theorem worked out on a small example, we refer the reader to Appendix~\ref{app:schutzen_example}. Although Theorem~\ref{thm:schutzdecomp} is sufficient to prove Sch\"utzenberger's theorem, the same inductive approach does not immediately lead to a quantum algorithm for star-free languages.  For example, it is not clear how to efficiently decide membership in $U \Sigma^{*}$ given an algorithm for membership in $U$.\footnote{We will show this is possible, but it requires that the language is regular. In general, a $\tilde{O}(\sqrt{n})$-query algorithm for a language $L$ does not imply a $\tilde{O}(\sqrt{n})$-query algorithm for $L \Sigma^{*}$. We have a counterexample: consider the language $L$ of strings of the form $\# x_0 \# x_1 \# x_2 \# \cdots \# x_k \#$ such that all $x_i$ are binary strings of the same length and $x_i = x_k$ for some $i < k$. $L$ can be decided in $\tilde{O}(\sqrt{n})$ queries by a Grover search. There is a clear reduction from element distinctness to $L \Sigma^{*}$, therefore $Q(L \Sigma^{*})$ is at least $\Omega(n^{2/3})$. } In the next section, we will strengthen our induction hypothesis such that queries of this type are possible.  Let us conclude this section with a splitting theorem based on Sch\"utzenberger's notion of rank.

\begin{theorem}
\label{thm:monoid_split}
Let $L = \varphi^{-1}(m)$ for monoid element $m \in M$.  Then, $L$ splits as 
$$
\bigcup_{pq = m} \varphi^{-1}(p) \varphi^{-1}(q).
$$
Furthermore, for all elements of the union, $\rho(p), \rho(q) \le \rho(m)$.
\end{theorem}
\begin{proof}
We first verify equality.  We have that $L \supseteq \cup_{pq = m} \varphi^{-1}(p) \varphi^{-1}(q)$ since 
$$
\varphi(\varphi^{-1}(p) \varphi^{-1}(q)) = \varphi(\varphi^{-1}(p)) \varphi(\varphi^{-1}(q)) = pq = m.
$$
Furthermore, 
$$
\bigcup_{pq = m} \varphi^{-1}(p) \varphi^{-1}(q) \supseteq \varphi^{-1}(m) \varphi^{-1}(1) = L.
$$
Now, suppose $x \in L$.  For any decomposition $x = uv$, we have that $\varphi(x) = \varphi(uv) = \varphi(u) \varphi(v) = m.$  Let $p = \varphi(u)$ and $q = \varphi(v)$.  Therefore, $u \in \varphi^{-1}(p)$ and $v \in \varphi^{-1}(q)$ with $pq = m$.  Finally, by Proposition~\ref{prop:rproduct} we get that $\rho(p), \rho(q) \le \rho(m)$.
\end{proof}

\subsection{\texorpdfstring{$\tilde{O}(\sqrt{n})$}{O(sqrt(n))} algorithm for star-free languages}
\label{sec:starfree_algorithm}

Recall that our objective is to create an $\tilde{O}(\sqrt n)$ algorithm for language $\varphi^{-1}(m)$, where $m \in M$ is an arbitrary monoid element. We mimic Sch\"utzenberger's proof of Theorem~\ref{thm:schutzenbergers_theorem} by constructing algorithms for each $\varphi^{-1}(m)$ in the order of the rank of $m$.  Implicit in such an argument is a procedure that must convert an efficient query algorithm for $\varphi^{-1}(r)$ into an efficient query algorithm for $\varphi^{-1}(r) a \Sigma^*$ for $(r,a) \in E$.

%This easily reduces to proving the induction hypothesis for $U \Sigma^{*}$, $\Sigma^{*} V$, $\Sigma^{*} C \Sigma^{*}$, and $\Sigma^{*} W \Sigma^{*}$, which in turn reduce to languages such as $\varphi^{-1}(n) a \Sigma^{*}$ for some $(n,a) \in E$. So, how do we prove our hypothesis for $\varphi^{-1}(n) a \Sigma^{*}$, given the same hypothesis for $\varphi^{-1}(n)$? 

%First, we need to strengthen the induction hypothesis. At a bare minimum, we want that membership in $\varphi^{-1}(m)$ can be decided in $\tilde{O}(\sqrt{n})$ quantum queries for $r(m) \leq n$. However, $\tilde{O}(\sqrt{n})$ membership tests for $L$ are not enough to decide membership in, for instance $L \Sigma^{*}$ (collision problem---see Theorem~\ref{TODO} ) in $\tilde{O}(\sqrt{n})$ queries. So, we strengthen our induction hypothesis to say we can test membership in $\varphi^{-1}(m) \Sigma^{*}$ (and by symmetry, $\Sigma^{*} \varphi^{-1}(m)$) in $\tilde{O}(\sqrt{n})$ queries, and therefore we can find the shortest prefix (or suffix) matching $\varphi^{-1}(m)$ in a string.

%Next, given the $\tilde{O}(\sqrt{n})$ algorithms to test membership in $L$ and $L \Sigma^{*}$, it is still not clear how to test membership in $L a \Sigma^{*}$ in $\tilde{O}(\sqrt{n})$ queries (modified collision---see Corollary~\ref{TODO}, for example), which is the situation for $\varphi^{-1}(n) a \Sigma^{*}$. 

Notice that for $(r, a) \in E$, we have (by definition) that $rM \supsetneq r \varphi(a) M$. That is, the prefix of the input string matching $\varphi^{-1}(r) a$ is not an arbitrary location in the string, but one of finitely many points in the string where the right ideal \emph{strictly} decreases. We use this to our benefit in the following key lemma. 
\begin{lemma}
\label{lem:prefixrightideal}
Let $\varphi \colon \Sigma^{*} \to M$ be a monoid homomorphism. Suppose there exists an $\tilde{O}(\sqrt n)$ membership algorithm for $\varphi^{-1}(m)$ for any $m \in M$ such that $\rho(m) \leq k$. Then, there exists an $\tilde{O}(\sqrt n)$ algorithm to test membership in $L := \varphi^{-1}(r) a \Sigma^{*}$ for any $r \in M$ and $a \in \Sigma$ such that $\rho(r) \leq k$ and $rM \supsetneq r \varphi(a) M$. 
\end{lemma}
\begin{proof}
Consider a string $x \in \Sigma^*$. The right ideal $\varphi(x_1 \cdots x_i) M$ represents the set of monoid elements we could reach after reading $x_1 \cdots x_i$. These right ideals descend as we read more of the string:
$$
M = \varphi(\eps) M \supseteq \varphi(x_1) M \supseteq \varphi(x_1 x_2) M \supseteq \cdots \supseteq \varphi(x_1 \cdots x_n) M = \varphi(x)M. 
$$
If $x \in L$, then there is some prefix $y$ in $\varphi^{-1}(r)$ followed by an $a$. By assumption, $\varphi(y) M = rM \supsetneq r \varphi(a) M = \varphi(ya) M$, so this is a point in the string where the right ideal strictly descends. 

Notice that $r \in r \varphi(a) M$ implies $rM \subseteq r \varphi(a) M$, and since we have $rM \supsetneq r \varphi(a) M$, we conclude that $r \notin r \varphi(a) M$. In other words, the right ideal descends from something containing $r$ (namely $rM$), to something not containing $r$ (namely $r \varphi(a) M$).

To decide whether $x$ belongs to $L$, it suffices to find the longest prefix $x_1 \cdots x_i$ such that $\varphi(x_1 \cdots x_i) M$ contains $r$. If $x_{i+1} = a$ and $x_1 \cdots x_{i} \in \varphi^{-1}(r)$, then the string is in $L$, otherwise there is no other possible prefix that could match $\varphi^{-1}(r) a$, so the string is not in $L$.
 
Define a new language $K$ where
$$
K := \bigcup_{s : r \in sM} \varphi^{-1}(s).
$$
This is precisely the language of strings/prefixes that could be extended to strings in $\varphi^{-1}(r)$. We can decide membership in $K$ with $O(\sqrt{n})$ queries because $r \in sM$ implies $MrM \subseteq MsM$ and hence $\rho(s) \leq \rho(r) \leq k$. 

It is also clear that $K$ is prefix closed: if $x_1 \cdots x_i \in K$ then $r \in \varphi(x_1 \cdots x_i) M \subseteq \varphi(x_1 \cdots x_{i-1}) M$, so $x_1 \cdots x_{i-1} \in K$ as well. The empty prefix is in $K$, and by binary search we can find the longest prefix in $K$. Then, as discussed above, we complete the algorithm by checking whether the prefix is (i) in $\varphi^{-1}(r)$ and (ii) followed by an $a$. If so, then we report $x \in L$, otherwise $x \notin L$. 
\end{proof}

We are now ready to state and prove our main theorem.
\begin{theorem}
\label{thm:main}
For any star-free language $L \subseteq \Sigma^{*}$, there exists a quantum algorithm which solves membership in $L$ with $\tilde{O}(\sqrt{n})$ queries and $\tilde{O}(\sqrt n)$ time. 
\end{theorem}
\begin{proof}
Let $L = \cup_{m \in S} \varphi^{-1}(m)$ for some homomorphism $\varphi \colon \Sigma^{*} \to M$ to an aperiodic finite monoid $M$, and $S \subseteq M$. We will show that there is an algorithm for each $\varphi^{-1}(m)$ by induction on the rank of $m$. 

%We argue by induction on $k$ that for all $m \in M$ such that $\rho(m) \leq k$, there exist $\tilde{O}(\sqrt{n})$-quantum query algorithms for $\varphi^{-1}(m)$, $\varphi^{-1}(m) \Sigma^{*}$, and $\Sigma^{*} \varphi^{-1}(m)$. 

Suppose first that $\rho(m) = 0$, implying that $m$ is the identity by Corollary~\ref{cor:shitty_corollary}. From Proposition~\ref{prop:unique_rank_zero_element}, we know that a string is in $\varphi^{-1}(1)$ if every character is in $\varphi^{-1}(1)$, i.e.,
$$
\varphi^{-1}(1) = \{ a \in \Sigma : \varphi(a) = 1 \}^{*}. 
$$
We can Grover search for a counterexample in $O(\sqrt{n})$ time to decide membership in $\varphi^{-1}(1)$.  

Now suppose $\rho(m)$ is nonzero. Our main tool is Theorem~\ref{thm:schutzdecomp}, which decomposes $\varphi^{-1}(m)$ into a Boolean combination of languages,
$$
\varphi^{-1}(m) = (U \Sigma^{*} \cap \Sigma^{*} V) \backslash (\Sigma^{*} C \Sigma^{*} \cup \Sigma^{*} W \Sigma^{*}), 
$$
where $U, V, C, W \subseteq \Sigma^{*}$ are as they appear in that theorem statement. We will also make reference to sets $E, F, G$ from Theorem~\ref{thm:schutzdecomp}. 

To give an algorithm for $\varphi^{-1}(m)$, it suffices to give an algorithm for each component of this Boolean combination: $U \Sigma^{*}$, $\Sigma^{*} V$, $\Sigma^{*} C \Sigma^{*}$ and $\Sigma^{*} W \Sigma^{*}$. Since $U$, $V$, and $W$ are finite unions of simpler languages, it suffices to consider each language in the union separately. 

The first component is $U \Sigma^{*}$, but we have already done most of the work for $U \Sigma^{*}$ in Lemma~\ref{lem:prefixrightideal}. Recall 
$$
U \Sigma^{*} = \bigcup_{(r,a) \in E} \varphi^{-1}(r) a \Sigma^{*}
$$
where $E = \{ (r,a) \in M \times \Sigma : r \varphi(a) M = mM, rM \neq mM \}$. This gives us an $\tilde{O}(\sqrt{n})$-time algorithm for $U \Sigma^{*}$. By symmetry, there also exists an algorithm for $\Sigma^{*} V$. Recall that $C = \{ a \in \Sigma : m \notin M \varphi(a) M \}$ is a finite set of characters, so membership in $\Sigma^{*} C \Sigma^{*}$ is decided by a Grover search for any of those characters.

The last component is $\Sigma^* W \Sigma^*$, which consists of a union of languages of the form $a \varphi^{-1}(r) b$ where $(a,r,b) \in G$. That is, $m \in M \varphi(a) r M$ and $m \in M r \varphi(b) M$ but $m \notin M \varphi(a) r \varphi(b) M$. We can use Theorem~\ref{thm:monoid_split} to split $W$ into 
$$
\bigcup_{pq = r} a \varphi^{-1}(p) \varphi^{-1}(q) b.
$$
We hope to apply Lemma~\ref{lem:prefixrightideal} to $\varphi^{-1}(q) b \Sigma^{*}$ and (in reverse) $\Sigma^{*} a \varphi^{-1}(p)$, then use infix search (i.e., Theorem~\ref{thm:scott_trick}) to try to find a substring in $W$, but first we need to verify that all the preconditions of these theorems are met---namely, that the rank of $p$ and $q$ are small, and $a$ and $b$ cause the ideal to descend.

First, the decomposition theorem (Theorem~\ref{thm:schutzdecomp}) gives that $\rho(r) < \rho(m)$, and by Proposition~\ref{prop:rproduct}, $\rho(p), \rho(q) \leq \rho(r)$. Next, suppose that $q \varphi(b) M = q M$. It follows that 
$$
M \varphi(a) r M = M \varphi(a) p q M = M \varphi(a) p q \varphi(b) M = M \varphi(a) r \varphi(b) M,
$$
but we know $m$ is in $M \varphi(a) r M$ and not in $M \varphi(a) r \varphi(b) M$, so we have a contradiction from the definition of $G$. Hence, $q \varphi(b) M \neq q M$, and by a symmetric argument $M \varphi(a) p \neq M p$, so we have $\tilde{O}(\sqrt{n})$-query algorithms for $\Sigma^{*} a \varphi^{-1}(p)$ and $\varphi^{-1}(q) b \Sigma^{*}$ from  Lemma~\ref{lem:prefixrightideal}. It follows that there is a $\tilde{O}(\sqrt{n})$ algorithm for $\Sigma^{*} W \Sigma^{*}$ as well.  
\end{proof}

This finishes the main theorem for this section. See Algorithm~\ref{alg:sf} for pseudocode.

\begin{algorithm}
	\caption{Star Free Language Algorithm}\label{alg:sf}
	\begin{algorithmic}
		\State{$\triangleright$ The monoid $M$, alphabet $\Sigma$, and homomorphism $\varphi \colon \Sigma^{*} \to M$ are fixed and known.}
		
		%\Function{GroverSearch}{\textit{range}, \textit{pred}}
		%	\State{$\triangleright$ Grover search over \textit{range} for an index satisfying predicate \textit{pred} in $O(\sqrt{n})$ calls to \textit{pred}.}
		%\EndFunction
		
		%\Function{BinarySearch}{\textit{range}, \textit{pred}}
		%	\State{$\triangleright$ Finds the last point $i \in \textit{range}$ such that $\textit{pred}(i)$ is true.}
		%\EndFunction
		%\Statex
				
		\Function{InfixSearch}{$x = x[1..n]$,$pred$}
		\State $\triangleright$ Searches for a substring matching the predicate \textit{pred}. See Theorem~\ref{thm:scott_trick}.
		\For{$\ell = 1,2,4,\ldots,n$}
		\State{Grover search for $i$ such that $\mathit{pred}(x[\min(1,i-\ell+1)..i], x[i+1..\max(i+\ell,n)]$ is true}
		\If{$i$ found}
		\State{\Return{\textsc{True}}}
		\EndIf
		\EndFor
		\State{\Return{\textsc{False}}}
		\EndFunction
		\Statex
		
		\Function{PrefixCheck}{$x$, $r$, $a$}
			\State $\triangleright$ \text{This function decides whether $x \in \varphi^{-1}(r) a \Sigma^{*}$ as described in Lemma~\ref{lem:prefixrightideal}.}
			\State $H \gets \{ s \in M : r \in sM \}$
			\State{Binary search for largest $1 \leq i < n$ satisfying $\bigvee_{s \in H} \textsc{Main}(x[1..i], s)$}
			\State \Return{$(x[i+1] = a) \wedge \textsc{Main}(x[1..i], r)$}
		\EndFunction
		\Statex
		
		\Function{RightIdeal}{$x$, $m$}
			\State{$\triangleright$ Checks if $x$ is in $U \Sigma^{*}$.}
			\State $E \gets \{ (r,a) \in M \times \Sigma : r \varphi(a) M = mM, rM \neq mM \}$
%			\If{$E = \varnothing$}
%				\State{\Return{\textsc{True}}}
%			\EndIf
			\For{$(r,a) \in E$}
				\If{\Call{PrefixCheck}{$x$, $r$, $a$}}
					\State{\Return \textsc{True}}
				\EndIf
			\EndFor
			\State{\Return \textsc{False}}
		\EndFunction
		\Statex
				
		\State{$\triangleright$ Define \textsc{SuffixCheck} and \textsc{LeftIdeal} likewise. Details omitted.}
		\Statex

		\Function{Ideal}{$x$,$m$}
			\State$\triangleright$ Checks if $x$ is in $\Sigma^{*} W \Sigma^{*}$. 
			\State $G \gets \{ (a,r,b) \in \Sigma \times M \times \Sigma : m \in (M \varphi(a) r M \cap M r \varphi(b) M) \backslash M \varphi(a) r \varphi(b) M \}$
			\For{$(a,r,b) \in G$}
			
				\If{\Call{InfixSearch}{$x$, $(x_1, x_2) \mapsto \bigvee_{pq = r} \textsc{SuffixCheck}(x_1, p, a) \wedge\textsc{PrefixCheck}(x_2, q, b)$}}
					\State{\Return{\textsc{True}}}
				\EndIf
			\EndFor
			\State{\Return{\textsc{False}}}
		\EndFunction	
		\Statex
		
%			\Function{IdealHelper}{$y$, $z$}
%		\For{$(p,q) : pq = r$}
%		\If{$\textsc{LeftIdeal}(x[i-\ell..i], p, a) \wedge\textsc{RightIdeal}(x[i+1..i+\ell], q, b)$}
%		\State{\Return{\textsc{True}}}
%		\EndIf
%		\EndFor
%		\State{\Return{\textsc{False}}}
%		\EndFunction

		\Function{\textsc{Main}}{$x = x[1..n]$, $m$}
		\State{$\triangleright$ Decides whether $x$ is in $\varphi^{-1}(m)$.}
		\If{$m = 1$}
			\State \Return{$\neg$ \Call{GroverSearch}{$\{1,\ldots,n\}$, $i \mapsto \varphi(x[i]) \neq 1$}}
		\Else
			\State \Return{$\left\{ \; \begin{matrix*}[l]
					\textsc{LeftIdeal}(x,m) &\wedge \\
					\textsc{RightIdeal}(x,m) &\wedge \\
					\neg \textsc{Ideal}(x,m) &\wedge \\
					\neg \textsc{GroverSearch}(\{1,\ldots,n\}, i \mapsto m \notin M \varphi(x[i]) M) &
				\end{matrix*}\right.$}

%			\Return{$\begin{matrix}
%					\textsc{LeftIdeal}{$x$,$m$} &$\wedge$ \\
%					\textsc{RightIdeal}{$x$,$m$} &$\wedge$ \\
%					$\neg$ \textsc{Ideal}{$x$,$m$} &$\wedge$ \\
%					$\neg$ \textsc{GroverSearch}{$\{1,\ldots,n\}$, $i \mapsto m \notin M \varphi(x[i]) M$} &
%				\end{tabular}}
%			\State $\begin{pmatrix} 0 & \ldots \\ \Sigma^{*} & 1 \end{pmatrix}$
								%\Call{RightIdeal}{x,m} \wedge$ 							$\neg \Call{Ideal}{x,m} \wedge							\neg $}
			\EndIf
		\EndFunction

	\end{algorithmic}
\end{algorithm}

%!TEX root = ../reg-query.tex

\section{Dichotomy Theorems}
\label{sec:dichotomies}

In this section, we prove a dichotomy result for block sensitivity. This will be important for the next logical step in the trichotomy theorem: proving lower bounds to match our upper bounds in Section~\ref{sec:lowerbounds}. The core of this section is a dichotomy theorem for sensitivity, namely that the sensitivity is either $O(1)$ or $\Omega(n)$.  This implies an identical dichotomy for block sensitivity, from which the $\Omega(\sqrt{n})$ lower bound on approximate degree follows for all nontrivial languages. 

%\begin{definition}
%Let $f \colon \Sigma^{*} \to S$ be a regular function. Let $x, y \in \Sigma^{*}$ be two strings which differ in exactly one position. If $f(x) \neq f(y)$ then we say the position where they differ is \emph{sensitive in $x$ (or $y$) with respect to $f$}. 
%
%Also, the \emph{sensitivity mask} of a string $x$ (with respect to a function $f$) is the binary string $m_x \in \{ 0, 1 \}^{*}$ which is $1$ at positions where $x$ is sensitive. 
%\end{definition}

Regular languages are closed under an astonishing variety of natural operations. Our $\Omega(\sqrt{n})$ lower bound begins with one such closure property. Recall that a symbol in a string is \emph{sensitive} with respect to some input $x$ if changing only that symbol changes the value of the function. 
\begin{theorem}
\label{thm:sensitivity_mask_is_regular}
Let $L \subseteq \Sigma^{*}$ be a regular language. Define the language $S_L \subseteq \{ 0, 1 \}^{*}$ of all \emph{sensitivity masks} as follows.
$$
S_L := \{ y \in \{ 0, 1 \}^{*} : \text{there exists $x \in \Sigma^{*}$ such that $|x| = |y|$ and $x_i$ is sensitive in $L$ if and only if $y_i = 1$} \}
$$
Then, $S_L$ is regular.
\end{theorem}
\begin{proof}
This is an exercise in using non-determinism, but since there are a few levels, let us spell out the details. First, let us show that the following language is regular:
$$
S'_L := \{ (x_1, y_1) \cdots (x_n, y_n) \in (\Sigma \times \{0,1\})^{*} : \text{$y_1 \cdots y_n$ indicate the sensitive bits of $x_1 \cdots x_n$} \}.
$$

How do we go about proving a string is \emph{not} in $S'_L$? There are two possibilities:
\begin{itemize} [itemsep=0pt]
\item Find some $i$ such that $y_i = 0$, but changing $x_i$ flips membership in the language.
\item Find some $i$ such that $y_i = 1$, but all possible changes to $x_i$ fail to flip membership in the language.
\end{itemize}
Each of these can be checked by a co-non-deterministic finite automaton. In the first case, we guess a position $i$ where $y_i = 0$, guess the new value of $x_i$, simulate the DFA on both paths and verify that they produce different outcomes. In the second case, we also guess a position $i$ where $y_i = 1$, but now simulate the original DFA for all possible values of $x_i$ and ensure that they are the same. Since there is a coNFA for $S'_L$, we get that $S_L$ is regular.

Now use non-determinism to reduce $S'_L$ to $S_L$: a string $y_1 \cdots y_n \in \{ 0, 1 \}^{*}$ is in $S_L$ if we can guess the accompanying $x_1 \cdots x_n \in \Sigma^{*}$ that puts it in $S'_L$. We conclude that there is an NFA accepting $S_L$, and therefore $S_L$ is regular.
\end{proof}

\begin{corollary}
\label{cor:sensitivity}
Let $L$ be a flat regular language. The sensitivity of $L$ is either $O(1)$ or $\Omega(n)$. 
\end{corollary}
\begin{proof}
Consider the language of sensitivity masks $S_L$ as defined in Theorem~\ref{thm:sensitivity_mask_is_regular}.  Notice that for a given length $n$, the sensitivity of $L$ is exactly the weight of the maximum Hamming weight string in $S_L$. Suppose the sensitivity is not $O(1)$.  Therefore, for any $k$, there exists a string $y_k \in S_L$ with Hamming weight at least $k$.

Since $S_L$ is a regular language, it has some pumping length\footnote{Let $L \subseteq \Sigma^*$ be a regular language.  There exists a finite pumping length $p > 0$ such that for all strings $w \in \Sigma^*$ with $|w| \ge p$ there exists a decomposition $w = xyz$ for $x,y,z \in \Sigma^*$ and $|y| > 0$, $w \in L \iff (\forall i \ge 0, xy^iz \in L)$.  This (or a similar statement) is called the ``pumping lemma'' since the substring $y$ may be repeated (``pumped") arbitrarily many times.} $p$. We can pump down any block of $p$ consecutive zero bits in $y_k$ such that at least $\frac{1}{p}$ fraction of the remaining bits are sensitive (or $n \leq kp$). This implies that sensitivity is $\Omega(n)$ for infinitely many $n$. We can also pump down arbitrary blocks of $p$ bits to decrease the length, so we can make sure sensitivity is $\Omega(n)$ for at least $\frac{1}{p}$ fraction of $n$. Finally, since $L$ is flat, congruence classes contain strings of all length, which allows us to replace some substring of a $\Omega(n)$ sensitive string with a slightly longer or shorter string. In this way, we can construct strings of sensitivity $\Omega(n)$ for all $n$.
\end{proof}

\begin{corollary}
\label{cor:block_sensitivity}
Let $L$ be a flat regular language. The block sensitivity of $L$ is either $O(1)$ or $\Omega(n)$.  
\end{corollary}
\begin{proof}
By Corollary~\ref{cor:sensitivity}, sensitivity is either $O(1)$ or $\Omega(n)$. If sensitivity is $O(1)$ then block sensitivity and all other measures are $O(1)$ by Corollary~\ref{cor:qcallconstant}. However, $s(f) \leq bs(f)$, so if sensitivity is $\Omega(n)$ then block sensitivity is $\Omega(n)$. It follows that block sensitivity is either $O(1)$ or $\Omega(n)$. 
\end{proof}
It follows that the certificate complexity, deterministic complexity, randomized zero-error complexity, randomized complexity are also $O(1)$ or $\Omega(n)$.

\begin{theorem}
\label{thm:approxdichotomy}
Let $L$ be a flat regular language. The approximate degree of $L$ is either $O(1)$ or $\Omega(\sqrt{n})$.\end{theorem}
\begin{proof}
Consider block sensitivity. If block sensitivity is $O(1)$, then so are approximate degree and quantum query complexity by Corollary~\ref{cor:qcallconstant}. If block sensitivity is $\Omega(n)$, then we recall that $\approxdeg(L) = \Omega(\sqrt{bs(L)}) = \Omega(\sqrt{n})$ by Theorem~\ref{thm:bs_relations}. Furthermore, $\frac{1}{2} \approxdeg(L) \leq Q(L)$ by Proposition~\ref{prop:qcinclusions}, so quantum query complexity is also $\Omega(\sqrt{n})$. 
\end{proof}

It follows that $Q(L)$ is either $O(1)$ or $\Omega(\sqrt{n})$.

%!TEX root = ../reg-query.tex

\section{Lower Bounds}
\label{sec:lowerbounds}

In this section, we will show matching lower bounds for the algorithms described in Section~\ref{sec:upperbounds}.  In fact, since approximate degree is a lower bound for quantum query complexity, it suffices to prove lower bounds for approximate degree, which is what we will do.  Let us start with simplest case---lower bounds on non-degenerate languages.

\begin{proposition}
Let $L$ be a flat regular language.  If $L$ is not degenerate, then $\approxdeg(L) \geq 1$.
\end{proposition}
\begin{proof}
Let $\varphi \colon \Sigma^* \to M_L$ be the homomorphism onto the syntactic monoid of $L$ such that $L = \{ \varphi^{-1}(s) : s \in S \subseteq M_L \}$.  Since $L$ is not degenerate, there exists $x, y \in \Sigma^+$ such that $\varphi(x) \neq \varphi(y)$.  By the definition of the syntactic congruence, there exist strings $u, v \in \Sigma^+$ such that $u \in L$ but $v \not \in L$.  Since $L$ is flat, each set $\varphi^{-1}(\varphi(u))$ and $\varphi^{-1}(\varphi(v))$ contains strings of all positive lengths.  Therefore, any polynomial approximating the membership function for $L$ cannot be constant.
\end{proof}

For the trivial languages, we first prove a theorem about their deterministic complexity.  Recall that a deterministic query algorithm is a decision tree:  on input $x \in \Sigma^n$, the algorithm queries a particular index of the input.  Based on the value of $x$ at that index (one of finitely many possible choices), the algorithm either deduces the membership of $x$ in $L$ or decides to query a different index.  The process is repeated until the algorithm can decide membership.  The height of the decision tree is the deterministic query complexity of $L$.  In particular, if the deterministic query complexity of $L$ is constant, then the height of the decision tree is constant, which implies that the entire tree has constant size (since each node in the tree has constant fan-out). 

\begin{theorem}
\label{thm:nontrivial_implies_nonconstant_d}
Let $L$ be a flat regular language. If $L$ is not trivial, then $D(L) = \omega(1)$. 
\end{theorem}
\begin{proof}
We will argue the contrapositive.  Suppose $D(L) = O(1)$.  That is, for any input $x \in \Sigma^n$, the deterministic algorithm queries a constant-size set of indices to determine membership.  Clearly, as $n$ increases, there will be large gaps between the indices which are queried.  Since $L$ is flat we have two important consequences: first, any nonempty string which is not queried can correspond to any non-identity element of the syntactic monoid; second, we may assume that any gap of nonzero size can be expanded or contracted to any other nonzero size. It follows that we can move the queries made by the deterministic algorithm (provided that we do not create or destroy any gaps) without changing its correctness.

Therefore, let us move all the queries as close to the start or end of the input as possible, maintaining $1$-symbol gaps where necessary. Since there are only constantly many queries, there exists a deterministic algorithm which determines membership of $x$ by querying $c$ symbols from the start and end of $x$ for some constant $c$.  

Let $\varphi$ be the homomorphism from $\Sigma^*$ onto the syntactic monoid $M_L$ such that $L = \{ \varphi^{-1}(s) : s \in S \subseteq M_L \}$.  For $x \in \Sigma^*$ of length greater than $2c$, write $x = u w v$ such that $|u| = |v| = c$.  We have that membership of $x$ in $L$ is determined completely by prefix $u$ and suffix $v$.  

We claim that this implies that $\varphi(u w v) = \varphi(u w' v)$ for all $w \in \Sigma^*$.  For suppose that $\varphi(u w v) \neq \varphi(u w' v)$.  By the definition of the syntactic congruence, there exists strings $a \in \Sigma^*$ and $b \in \Sigma^*$ such that $a u w v b \in L$ and $a u w' v b \not \in L$ (or vice versa).  Since $|au| > 0$ and $|bv| > 0$, there exists strings $a_u, b_v \in \Sigma^c$ such that $\varphi(a_u) = \varphi(au)$ and $\varphi(b_v) = \varphi(bv)$.  However, $a_u w b_v \in L$ and $a_u w' b_v \not \in L$ contradicts the fact that membership in $L$ is determined by a prefix and suffix of length at most $c$.  In particular, this holds when $w' = \varepsilon$.

Let us now show that  $\varphi(\Sigma^+)$ is a rectangular band.  Let $x,y,z \in \Sigma^+$ be nonempty strings, and let $x',z'$ be strings of length $c$ such that $\varphi(x) = \varphi(x')$ and $\varphi(z) = \varphi(z')$.
We have that 
$$
\varphi(xyz) = \varphi(x' y z') = \varphi (x' z') = \varphi(xz).
$$
Finally, we show that $\varphi(\Sigma^+)$ is idempotent.  Let $x \in \Sigma^+$.  By flatness, we have that $\varphi(x) = \varphi(a w b)$ for strings $a, b \in \Sigma^c$.  Therefore, we have
$$
\varphi(x) = \varphi(a w b) = \varphi(a b) = \varphi(a b a b) = \varphi(xx),
$$
where the middle two equalities come from substituting $w = \varepsilon$ and $w = ba$, respectively.
\end{proof}

\begin{corollary}
Let $L$ be a flat regular language.  If $L$ is not trivial, then $\approxdeg(L) = \Omega(\sqrt n)$.
\end{corollary}
\begin{proof}
The corollary follows almost immediately from Theorems~\ref{thm:approxdichotomy} and \ref{thm:nontrivial_implies_nonconstant_d}. Suppose $\approxdeg(L) = o(\sqrt n)$.  We wish to show that $L$ is trivial.  If $D(L) = O(1)$, then we are done by Theorem~\ref{thm:nontrivial_implies_nonconstant_d}. If $D(L) = \omega(1)$, then approximate degree is also non-constant by Corollary~\ref{cor:qcallconstant}. But if $\approxdeg(L)$ is non-constant, then we must have $\approxdeg(L) = \Omega(\sqrt n)$ by Theorem~\ref{thm:approxdichotomy}.
\end{proof}

Finally, we turn our attention to the star-free languages.  Let $\MOD_p$ be the language of bit strings whose Hamming weight is 0 modulo some fixed $p \geq 2$. We need the following theorem:
\begin{theorem}[Beals et al.\ \cite{beals:2001}]
\label{thm:parityishard}
$\approxdeg(\MOD_p) = \Omega(n)$ for any $p \geq 2$. 
\end{theorem}
Recall that star-free languages are aperiodic.  Therefore, if a language is not star free, then it should exhibits some periodicity in which we can embed some $\MOD_p$ language.   We appeal to this intuition in the following theorem.
\begin{theorem}
Let $L$ be a flat regular language. If $L$ is not star free, then $\approxdeg(L) = \Omega(n)$. 
\end{theorem}
\begin{proof}
Let $M_L$ be the syntactic monoid of $L$, and let $\varphi \colon \Sigma^{*} \to M_L$ be the accompanying surjection onto $M_L$. We assume $M_L$ is not aperiodic, so there exists an element $s \in M_L$ such that $s^{n} \neq s^{n+1}$ for any $n$. Since $M_L$ is finite, we have $s^{n} = s^{n+p}$ for some $p$ and $n$, and therefore for all sufficiently large $n$. Let us take the minimal $p$ so that $s^{n} \neq s^{n+i}$ for $0 < i < p$. 

Since the language is flat, there exist $a_0, a_1, b \in \Sigma$ such that $\varphi(a_0) = s^{p}$, $\varphi(a_1) = s$ and $\varphi(b) = s^{n}$. One might worry that if $s^{n}$ is equal to the identity, its only preimage is the empty string, as is sometimes true for flat languages. However, because $\varphi(a_1^{n}) = \varphi(a_1)^{n} = s^n$, this is not the case. Given string $x \in \{ 0, 1 \}^{m}$, observe that
$$
\varphi( a_{x_1} a_{x_2} \cdots a_{x_m} b) = s^{x_1 + x_2 + \cdots + x_m + n},
$$
since $s^{n+p} = s^{n}$.  In other words, the monoid element associated with $a_{x_1} \cdots a_{x_m} b$ is determined by the Hamming weight of $x$ modulo $p$.  Therefore, to decide membership of $x$ in $\MOD_p$, it suffices to compute the monoid element for $a_{x_1} \cdots a_{x_m} b$ in $M_L$. 

Finally, by the definition of syntactic congruence, any two monoid elements may be distinguished by prepending and appending fixed strings to the input, then testing membership in $L$. By flatness, we may take those strings to be length zero or one. Thus, we can determine the monoid element by a constant number of queries to $L$, and therefore compute the Hamming weight modulo $p$. It follows that membership in $L$ has approximate degree $\Omega(n)$ by Theorem~\ref{thm:parityishard}. 
\end{proof}

%!TEX root = ../reg-query.tex

\section{Context-Free Languages}
\label{sec:cfl}

In this section we will prove that the context-free languages---a slightly larger class of languages containing the regular languages---have query complexities outside the trichotomy.  The context-free languages are most often defined either through context-free grammars or through pushdown automata (PDA).  It will be easier for us to work with the PDA definition in this section.

One can think of a PDA as a nondeterministic Turing machine which has a read-once input tape and read-write stack.  Although the addition of the stack allows PDA to recognize many languages which are not regular, they are still limited in many senses.  For instance, context-free languages exhibit a pumping lemma much like the regular languages, and the membership problem is decidable.  For a more formal definition we refer the reader to introductory texts \cite{sipser:2006}.

As a simple example, consider the Dyck language over alphabet $\Sigma = \{ \lparen, \rparen \}$, which consists of all words with balanced parentheses.  We can show this language is context free by constructing a PDA for it.  The idea is that the stack contains all of the unmatched left parentheses.  When a new parenthesis is read from the input tape, the PDA pushes it onto the stack if it is a left parenthesis or pops an item from the stack if there is a right parenthesis.  The PDA accepts if the stack is empty when the input is read entirely.  

\subsection{Context-free languages do not obey the trichotomy}
In general, the easiest way to construct a language with arbitrary query complexity is by padding a hard language. The procedure is simple:  take a problem with $\Omega(n)$ query complexity, e.g., parity, and make the input string longer by adding (or \emph{padding}) irrelevant symbols to the end. For instance, computing the parity of the first $\Theta(n^{2/3})$ bits and ignoring the rest will require $\Theta(n^{2/3})$ queries. 

Unfortunately, to create a \emph{context-free} language with arbitrary query complexity, we cannot take such a direct approach.  Context-free languages cannot simply count out some fraction of their input as the above example suggests.  Instead, let us consider a general procedure for constructing a context-free language $L \subseteq \Sigma^*$ which has quantum query complexity $\Theta(n^c)$ for some $c \in [1/2, 1]$.  We construct $L$ from the union of two context-free languages $A$ and $B$.  To test membership of some $x \in \Sigma^*$ in $L$, we first test whether or not $x$ belongs to $A$.  We always construct $A$ in such a way that membership in $A$ can be decided in $O(\sqrt n)$ queries, usually through a simple Grover search.\footnote{In fact, the reason we cannot extend this procedure to other exponents $c \in (0,1/2)$ is due to the fact that we will always incur this cost of Grover search.}  If $x \in A$, then we are done.  Otherwise, we can assume that $x \not\in A$ when testing membership in $B$.  However, $A$ is constructed such that $x \not \in A$ will imply that $x$ has been ``padded''---there is some special symbol in $x$ such that the distance from that symbol to the beginning of the string is approximately $n^c$.  Therefore, if $B$ is the language of all strings such that the prefix before the special symbol has even parity, then the query complexity of $L = A \cup B$ is $\Theta(n^c)$.

Let us consider an example of such a language $A \subseteq (\Sigma \cup \{\hashtag\})^*$.  First, we enforce that every word in $A$ begins and ends with $\hashtag$.  Next, we say that $x \in A$ iff there is some substring $\hashtag y \hashtag$ of $x$ such that $y \in \Sigma^*$ and the length of $y$ is \emph{not} equal to the total number of $\hashtag$ symbols in $x$.  Notice that  $x \not \in A$ implies that $x = \hashtag x_1 \hashtag x_2 \hashtag \ldots \hashtag x_k \hashtag$ where $|x_i| \approx \sqrt n$.  Furthermore, $A$ is context free and the quantum query complexity of $A$ is $\Theta(\sqrt n)$ by Grover search.

We will prove a theorem vastly generalizing this approach to create substrings of length $n^c$ for any $c \in [1/2,1]$ which is limit computable.\footnote{Since the theorem constructs a very contrived language, we note that natural problems can also be embedded into context-free languages, e.g., the element distinctness problem.  Given a list of integers $x_1, \ldots, x_n$ such that each $x_i \in \{1, \ldots, m\}$, the element distinctness problem asks if there exists $i \neq j$ such that $x_i = x_j$. Since $m \ge n$, we write each $x_i$ as a string over $\{ 0, 1 \}$, and delimit the $x_i$s by $2$'s. The language $\CFED$ consists of grammar rules: $
S \to A 2 B 2 A,
B \to 0 B 0 \mid 1 B 1 \mid 2 A 2,
A \to 0 \mid 1 \mid 2 \mid \varepsilon \mid AA
$. $\CFED$ accepts strings where some $x_i$ is the \emph{reverse} of some $x_j$. Thus, if all $x_i$ are represented by palindromes, $\CFED$ is at least as hard as element distinctness. On the other hand, it is possible to adapt the algorithm $O(n^{2/3})$ quantum query algorithm for element distinctness to $\CFED$ (with a log factor loss)\cite{ambainis:2005, kutin:2005}.}  A number $c \in \mathbb{R}$ is \emph{limit computable} if there exists a Turing machine which on input $n$ outputs some rational number $T(n)$ such that $\lim_{n \to \infty} T(n) = c$.  

We will need two main technical lemmas, both of which define a language similar to $A$ above.  The first ensures that the input contains (as a substring) the total length of the input written in binary, and the second simulates arbitrary computation by a Turing machine.

\begin{lemma}[Proof in Appendix~\ref{sec:contextfreeproofs}]
\label{lem:cflcounting}
	Let $K \subseteq \{ 0, 1, \hashtag_1, \hashtag_2, \$ \}^*$ be the language such that
	\begin{itemize}[itemsep=0pt]
		\item if $x \in K$, then $x$ ends with $\$ y \hashtag_1$, and
		\item for all $n \geq 6$, there is an $x \in K$ ending in $\$ y \hashtag_1$,
	\end{itemize}
	where $y$ is the binary representation of $|x|$. Then, $\overline{K}$ is context free, and $Q(K) = O(\sqrt{n})$.  
\end{lemma}

\begin{lemma}[Proof in Appendix~\ref{sec:contextfreeproofs}, folklore \cite{sipser:2006}]
\label{lem:cflhistory}
Let $N$ be a $k$-tape nondeterministic Turing machine.  Define language $K_N$ which contains strings of the form
$$
C_1 \hashtag C_2^R \hashtag C_3 \ldots C_{n-1}^R \hashtag C_n
$$
where $C_1$ is a valid start configuration of $N$, $C_n$ is a valid accepting configuration, and $C_i$ to $C_{i+1}$ is a valid transition. Then, $\overline{K_N}$ is context free, and $Q(K_N) = O(\sqrt n)$.
\end{lemma}

We are now ready to construct context-free languages that have quantum query complexities corresponding to limit computable exponents.  Although there are several technical details to check in the proof, the central idea is straightforward:  Let $x \in \Sigma^n$ be the input.  If $x$ is not in the language defined in Lemma~\ref{lem:cflcounting}, then the $n$ will be written in binary on the string.  If $x$ is also not in the language defined in Lemma~\ref{lem:cflhistory}, then the input will contain a correct simulation of a Turing machine limit computing some query exponent $c \in [1/2,1]$ and verifying that a $\hashtag$ symbol has been placed at position $n^c$.  Using Grover search, we can verify that the membership in these languages in $O(\sqrt n)$ time.  If $x$ is in neither language, then computing parity on the prefix of the input (up to the $\hashtag$ symbol) takes time $\Theta(n^c)$, from which the theorem follows. 

\begin{reptheorem}{thm:cfl_computable}
For all limit computable $c \in [1/2,1]$, there exists a context-free language $L$ such that $Q(L) = O(n^{c + \epsilon})$ and $Q(L) = \Omega(n^{c - \epsilon})$ for all $\epsilon > 0$.  Furthermore, if an additive $\epsilon$-approximation to $c$ is computable in $2^{O(1/\epsilon)}$ time, then $Q(L) = \Theta(n^{c})$. In particular, any algebraic $c \in [1/2, 1]$ has this property. 
\end{reptheorem}
\begin{proof}
Let $M$ be the Turing machine computing $c$ in the limit. That is, on input $1^{k}$ it outputs a rational approximation $c_k$ such that $\lim_{k \to \infty} c_k = c$. Let $n_k$ be the size of the computation history when computing $c_k$. Without loss of generality, we assume $n_k$ is strictly increasing with $k$. We also assume that the computation history for computing $n^{c_k}$ from $n$ and $c_k$ (both written in binary) is of size at most $n$ for all $n \geq n_k$.\footnote{Exponentiation will run in polynomial time, which is actually $\polylog(n)$ since $n$ is written in binary, plus the description size of $c_k$. In fact, we can be even sloppier and approximate $n^{c_k}$ with the first $c_k$ fraction of the bits of $n$, and still be accurate up to constant factors.}

Our goal is to construct a context-free language which accepts any input not satisfying the following array of conditions. Note that each condition may require a complicated witness to verify, perhaps as long as the input itself.  Therefore, we let the alphabet be tuples $\Sigma = \Sigma_1 \times \Sigma_2 \times \cdots \times \Sigma_m$ so that there are $m$ independent \emph{tracks} to work with. Suppose the input has length $n$, and consider the following six tracks.
\begin{enumerate}[itemsep=0pt]
	\item The first track contains bits and a $\$$ symbol, hopefully at position $\lceil n^{c_k} \rceil$ in the string.
	\item Some Turing machine $M$ limit computes $c \in [1/2, 1]$. The second track holds a valid computation history of $M$ computing some $c_k$ from input $1^k$. 
	\item The third track contains an \emph{incomplete} execution of $M$ on $1^{k+1}$. If the string is long enough to complete the computation of $c_{k+1}$, then $c_k$ is obsolete and should not be used. 
	\item The fourth track contains a binary number (and associated machinery, see Lemma~\ref{lem:cflcounting}) matching the length of the input.
	\item The fifth track is the same as the fourth, except the number is the position of $\$$ on track one.  
	\item The sixth track holds a Turing machine computation history which verifies that the position of $\$$ is $\lceil n^{c_k} \rceil$, based on the numbers from tracks $2$ ($c_k$), $4$ ($|z|$), and $5$ (position of $\$$). 
\end{enumerate}
We enforce these conditions with a corresponding array of context-free languages, which \emph{reject} satisfying strings. The final language will be a union of these languages, so that it rejects a string if and only if all the conditions are satisfied. 

We have already seen most of the languages we need. For example, we want to accept if track two is not the computation history of the Turing machine $M$ which computes $c_k$, but we have seen how to construct such a language in Lemma~\ref{lem:cflhistory}. Similarly for tracks three and six, we can tweak this construction to accept on \emph{incomplete} computation histories. On track four we want precisely the binary counter language in Lemma~\ref{lem:cflcounting}. Track five is the same thing concatenated with $\$ \Sigma^{*}$ to ignore symbols after $\$$. Track one is actually just a regular language, $\Sigma_1^{*} \backslash (0|1)^{*} \$ (0|1)^{*}$.

Each language so far focuses on just one of the tracks, and we need a few ``glue" languages to ensure the various tracks interact correctly. The first verifies that $c_k$, $n$, and the position of $\$$ (appearing on track two, track four, and track five respectively) match the strings in the input configuration on track six. A second glue language checks that if the starting configuration on track two was $1^{k}$, then the starting configuration on track three is $1^{k+1}$, so that it computes $c_{k+1}$. The third and final glue language checks that the $\$$ on track one matches the $\$$ on track five. We arrange for all of the glue languages to accept strings which fail these checks, in keeping with the complemented behavior of the other languages.

Suppose we have a string of length $n_k \leq n < n_{k+1}$ which is \emph{rejected} by all of the languages. It follows that all of the conditions are satisfied, so we can say a lot about the string. First, track two must compute $c_{k'}$ for some $k' \leq k$, since there is not enough space for $c_{k+1}$. Similarly, track three stops in the middle of computing $c_{k'+1}$ for some $k' \geq k$, since for small $k'$ it would have finished. But $k'$ is the same in both cases (due to a glue language), and hence $k' = k$. Track four and five generate binary numbers for the length and position of the $\$$ symbol which, by another glue language, are written on the input of track six, along with $c_k$. Finally, the sixth track verifies that the position is indeed $\lceil n^{c_k} \rceil$.

We have one final language, which accepts depending on the parity of the bits on the first track up to the $\$$. If all of the other languages reject, then we have argued that $\$$ is at position $\lceil n^{c_k} \rceil$, so computing the parity takes $n^{c_k}$ queries, up to constant factors. Checking all the other conditions takes $O(\sqrt{n})$ queries, so the cost of computing the parity dominates because $c_k \geq \frac{1}{2}$. It follows that the quantum query complexity is within a constant factor of $n^{c_k}$ for $n$ between $n_k$ and $n_{k+1}$. For any $\epsilon > 0$ we have $| c - c_k | \leq \epsilon$ for sufficiently large $k$, and hence for sufficiently large $n$, the query complexity is $Q(L) = O(n^{c + \epsilon})$ and $\Omega(n^{c - \epsilon})$. 

Finally, we note that if the $c_k$s converge sufficiently quickly (with respect to computation time, not $k$) then the query complexity is truly $Q(L) = \Theta(n^{c})$. For example, suppose we have a Turing machine which spends $2^{O(1/\epsilon)}$ time to output an approximation $c'$ to $c$ such that $|c - c'| \leq \epsilon$, for any $\epsilon > 0$. It does not matter whether the machine outputs a stream of better and better approximations, or takes $\epsilon$ as input and outputs a sufficiently good approximation. Either way, we can construct a machine which maps $1^{k}$ to $c_k$ with a similar guarantee: the time to compute $c_{k+1}$ is at most $2^{O(1/\epsilon_k)}$ where $\epsilon_k = |c - c_k|$. We claim this is enough to show $Q(L) = \Theta(n^{c})$. 

Our construction of $L$ is such that the query complexity is (up to constant multiplicative factors) $n^{c_k}$ on the entire interval $[n_k,n_{k+1})$. For convenience, define a function $c(n) \colon \mathbb{N} \to \mathbb{R}$ such that $c(n) = c_k$ iff $n \in [n_k,n_{k+1})$. This means $Q(L) = \Theta(n^{c(n)})$, and taking logs gives $\left|\frac{\log Q(L)}{\log n} - c(n) \right| = O(1/\log n)$. Recall $n_{k+1} \le 2^{b/\epsilon_k}$ for some $b$ (and all sufficiently large $k$) so we have $|c - c_k| = \epsilon_k \leq b/\log n_{k+1}$. It follows that $|c - c(n)| \leq O(1/\log n)$. Together, this implies 
$$
\left| \frac{\log Q(L)}{\log n} - c \right| \leq \left|\frac{\log Q(L)}{\log n} - c(n) \right| + |c(n) - c| \leq \frac{a}{\log n}
$$
for some $a$ and for all sufficiently large $n$. It follows $2^{-a} n^{c} \leq Q(L) \leq 2^{a} n^{c}$, so we conclude that for sequences converging sufficiently quickly, $Q(L) = \Theta(n^{c})$. For example, any algebraic number can be computed to $1/\epsilon$ precision in $\polylog(1/\epsilon)$ time using Newton's method from a suitable starting point.
\end{proof}

The converse of this theorem also holds.  

\begin{reptheorem}{thm:cfl_converse}
Let $L$ be a context-free language such that $\lim_{n \to \infty} \frac{\log Q(L)}{\log n} = c$.  Then, $c$ is limit computable.
\end{reptheorem}
\begin{proof}
Suppose $L$ is context free.  Recall that given $w \in \Sigma^*$, the problem of computing membership of $w$ in $L$ is decidable \cite{sipser:2006}.  Next, we observe that the quantum query complexity can be expressed as the solution (up to logarithmic factors) to a large semi-definite program \cite{reichardt:2009}.  That is, there exists a Turing machine which outputs $\adv(L)$ such that $Q(L) = \tilde{\Theta}(\adv(L))$.  Therefore, we can construct a Turing machine which outputs $\log(\adv(L))/(\log n)$, and
$$
\lim_{n \to \infty} \frac{\log(\adv(L))}{\log n} = \lim_{n \to \infty} \frac{\log \tilde{\Theta}{(Q(L))}}{\log n} = \lim_{n \to \infty} \frac{\log Q(L)}{\log n} = c.
$$
Therefore, $c$ is limit computable.
\end{proof}

%!TEX root = ../reg-query.tex

\section{Future Work}
\label{sec:futurework}

Recall that the $\tilde{O}(\sqrt n)$ algorithm for star-free languages incurs \emph{many} log factors.  This suggests a natural question:  what is the exact upper bound for the query complexity for star-free languages?  Proving that even one log factor is necessary seems challenging.

Next, we are interested in extending the hierarchy to other languages and settings. The context-free languages, for example, seem like a natural step. We know (see Section~\ref{sec:cfl}) that there is no longer a trichotomy; for every limit computable number $c \in [1/2, 1]$, there exists a context-free language with quantum query complexity approaching $\Theta(n^c)$. We also conjecture that no context-free language has quantum query complexity $\omega(1)$ but also $o(\sqrt{n})$.

Another setting to consider is promise problems. In this work, we required the query algorithm to decide membership on all strings. If we restrict the input strings to some promise set, it may affect the query complexity.
Allowing for an arbitrary promise trivially leads to languages with quantum query complexity $\Theta(f(n))$ for an arbitrary function $f$ between $0$ and $n$.  For example, consider the parity function with the promise that only the first $f(n)$ bits are nonzero. 
%Aaronson and Ben-David \cite{aaronson:2016} show that for quite general problems, the query complexity can be sculpted with a judicious choice of promise. 
 Instead, let us take the promise to be a regular language.  In this model, we can construct a binary search language with query complexity $\Theta(\log n)$.  Formally, the problem is to decide whether there is an occurrence of $01$ at an even position (i.e., membership in $(\Sigma \Sigma)^{*} 01 \Sigma^{*}$) promised that the input is sorted (i.e., belongs to $0^{*} 1^{*}$).  We conjecture that the trichotomy becomes $\Theta(\polylog(n))$, $\Theta(\sqrt{n} \cdot \polylog(n))$, or $\Theta(n)$.
 
We are interested in more applications of the star-free algorithm. For example, in the classical world, linear-time algorithms for the string matching problem have been derived from finite automata. Quantum algorithms for string matching with quadratic speedup are known \cite{ramesh:2003}, but can one derive a quadratic speedup by applying our algorithm for star-free languages as a black box? As a toy example, notice that for any fixed $w$, the language $\Sigma^{*} w \Sigma^{*}$ is star free, so we obtain $\tilde{O}(\sqrt n)$ string search for fixed queries. 

Finally, consider the restricted Dyck language introduced in Section~\ref{sec:applications}---the language of nested parentheses where the parentheses are only allowed to nest $k$ levels deep.  When $k$ is constant, this language is star free and therefore has quantum query complexity $\tilde{\Theta}(\sqrt n)$.  When $k$ is unbounded, consider the set of inputs $w_x = \lparen \lparen \ldots \lparen x  \rparen \ldots \rparen \rparen$ where $x \in \{\lparen, \rparen\}^{n/3}$ and there are exactly $n/3$ leading left parentheses and $n/3$ trailing right parentheses.  Notice that $w_x$ is the Dyck language iff the number of left parentheses in $x$ is equal to the number of right parenthesis in $x$.  Therefore, the quantum query complexity for $k = \Omega(n)$ is $\Omega(n)$.  We now ask the question: what is the quantum query complexity of the restricted Dyck language when $k$ is sublinear but superconstant?

%Recall that star-free languages are related to $\AC^{0}$, which is closely related to the polynomial hierarchy. Does our algorithm for testing star-free languages imply any interesting complexity-theoretic consequences for the polynomial hierarchy, in the same sense that Grover search implies better-than-brute-force algorithms for problems in $\NP$, falsifying at least one version of quantum SETH?
 
% TODO - communication complexity

%!TEX root = ../reg-query.tex

\section{Acknowledgements}

We thank Andris Ambainis, Shalev Ben-David, Robin Kothari, Han-Hsuan Lin, and Ronald de Wolf for useful discussions.

\bibliographystyle{plain}
\bibliography{bibliography}

\appendix
%!TEX root = ../reg-query.tex

\section{Flattening Details}
\label{sec:flattening}

Let us start with a precise definition of what we mean in this paper by a flat regular language.

\begin{definition}
	Let $\varphi \colon \Sigma^{*} \to M$ be a monoid homomorphism onto a finite monoid. Let $\Sigma_k$ denote the non-empty strings of length divisible by $k$. The \emph{conductor} is the least integer $K$ such that $\varphi(\Sigma_K) = \varphi(\Sigma^{nK})$ for all $n \geq 1$.
	A regular language $L \subseteq \Sigma^{*}$ recognized by the morphism $\varphi \colon \Sigma^* \to M_L$ onto its syntactic monoid is \emph{flat} if its conductor is $1$. 
\end{definition}
Once we convert the language to blocks of size $K$ (i.e., alphabet $\Sigma^{K}$), any congruence class of the monoid containing a non-empty string contains strings of all (non-zero) lengths. We refer to this as Property~\ref{property:flat} in Section~\ref{sec:formalstatement}. However, we still need to show $K$, and therefore flat regular languages, exist. 
\begin{theorem}
	For any homomorphism $\varphi \colon \Sigma^{*} \to M$ onto a finite monoid, the conductor is finite and computable.
\end{theorem} 
\begin{proof}
	Let $\lambda \colon \Sigma^{*} \to \mathbb N$ be the homomorphism mapping strings to their lengths. The set $A_r := \lambda(\varphi^{-1}(r))$ is ultimately periodic, i.e., there exists $p$ such that $A_r$ and $A_r+p$ differ at finitely many points. This may be easier to see by mapping $\varphi^{-1}(r)$ to unary, and since the language is still regular, considering the DFA. Let $K'$ be the least common multiple of the period of $A_r$ for all $r \in M$. We will take $K$ to be a multiple of $K'$, so we may as well assume without loss of generality that the period of $A_r$ is $1$.
	
	When a set of natural numbers has period $1$, it is either finite or cofinite. Take $K$ larger than all the finite exceptions in either case. That is, for all $r$, take $K$ larger than the maximum element in $A_r$ (if finite) and the maximum element not in $A_r$ (if cofinite). The result is that each $A_r \cap K \mathbb N$ is one of $\varnothing$, $\{ 0 \}$, $K \mathbb N$, or $K \mathbb N \backslash \{ 0 \}$. Only the identity class, $A_1$, can contain $0$, so all other $A_r$ are either $\varnothing$ or $K \mathbb N \backslash \{ 0 \}$. We throw away $r \in M$ such that $A_r = \varnothing$, and the remaining elements have the property, by construction, that they are the images of strings of all lengths divisible by $K$. 
\end{proof}

We are finally ready to restate and prove Theorem~\ref{thm:flattening}, which states that any regular language can be divided into a collection of flat languages.

\begin{reptheorem}{thm:flattening}
Let $L \subseteq \Sigma^{*}$ be a regular language recognized by a monoid $M$. There exists an integer $p \geq 2$ and a finite family of flat regular languages $\{ L_i \}_{i \in I}$ over alphabet $\Sigma^{p}$ such that testing membership in $L$ reduces (in fewer than $p$ queries) to testing membership in some $L_i$. Furthermore, the same monoid $M$ recognizes $L$ and every $L_i$.
\end{reptheorem}
\begin{proof}
Let $p$ be the conductor of $L$.  Consider an input $x \in \Sigma^{*}$ of length $n$. Clearly we can divide $x$ into a string $x' \in (\Sigma^p)^{*}$ of length $\lfloor n/p \rfloor$, and a remainder $r \in \Sigma^{*}$ of length less than $p$. For each such $r$, we define the language 
	$$
	L_r := \{ y \in (\Sigma^{p})^{*} : yr \in L \},
	$$
	slightly abusing notation so that $y$ denotes both a string over $\Sigma^{p}$ and a string over $\Sigma$. We leave it as an exercise to show that $L_r$ is regular. By construction, $x$ is in $L$ if and only if $x'$ is in $L_r$, so by looking at length of the input and the last $|r|$ symbols, we have reduced testing membership in $L$ to membership in $L_r$. 
	
	Finally, let $\varphi^{p} \colon (\Sigma^{p})^{*} \to M$ denote the extension of $\varphi$ to strings over $\Sigma^{p}$. Note that we can write $L_r$ as 
	\begin{align*}
	L_r &= \{ y \in (\Sigma^{p})^{*} : \varphi^{p}(y) \varphi(r) \in S \} \\
	&= \{ y \in (\Sigma^{p})^{*} : \varphi^{p}(y) \in \{ q \in M : q \varphi(r) \in S \} \} \\
	&= (\varphi^{p})^{-1}( \{ q \in M : q \varphi(r) \in S \}).
	\end{align*}
	It follows that $L_r$ is recognized by $M$. By construction, the conductor of $L_r$ is 1, so $L_r$ is flat.
\end{proof}

\subsection{Monotonic query complexity}
\label{sec:monotone_query_complexity}

Let us now consider an alternative to flattening---namely, modifying the definition of query complexity so that it is nondecreasing.  For this section only, define the quantum query complexity $Q(f)(n)$ of function $f$ to be the minimum number of quantum oracles calls needed to determine the value of $f$ on all strings of length \emph{up to} $n$.  When query complexity is defined in this way, we can prove a quantum query complexity trichotomy theorem for \emph{all} regular languages as a corollary of our trichotomy theorem for flat languages.

%\begin{proposition}
%For flat regular languages the two notions of quantum query complexity coincide.
%\end{proposition}

\begin{theorem}
Let $L \subseteq \Sigma^*$ be any regular language.  The quantum query complexity of $L$ is either $0, \Theta(1), \tilde{\Theta}(\sqrt n),$ or $\Theta(n)$.
\end{theorem}
\begin{proof}
By Theorem~\ref{thm:flattening}, we have that $L$ is a finite disjoint union of languages $L_r r$ where each $r \in \Sigma^{*}$ has length less than $p$.  Technically, $L_r$ is over the alphabet $\Sigma^p$, but we extend strings in $L_r$ to strings over alphabet $\Sigma$ in the obvious way.  If all $L_r$ have constant query complexity, then $Q(L) = 0$ or $Q(L) = \Theta(1)$.  Therefore, assume there is some $L_r$ such that $Q(L_r) = \omega(1)$.  We will show that $Q(L) = \Theta(\max_r Q(L_r))$.

Let us consider one algorithm for $L$ on strings of length $np + i$ where $i < p$: query the last $i$ characters of the string to determine $r$, and then use at most $p Q(L_r)(n)$ queries to test the rest.  Therefore, we have\footnote{We now see the need to separate the constant and non-constant cases.  The additive $p$ factor would technically take a 0-query algorithm to an $\Theta(1)$-query algorithm, which we want to avoid.}
$$
Q(L)(np + i) \le \max_{r} p Q(L_r)(n) + p.
$$

In the other direction, notice that by decreasing the length of the string by at most $p$, we can have any remainder string $r$.  By the modified definition of query complexity, shortening the length must decrease the query complexity.  Since we can force the query algorithm to solve any smaller instance of a flat language $L_r$, we have
$$
Q(L)(np+i) \ge \max_{r} Q(L_r)(n-1).
$$
That is, $Q(L) = \Theta(\max_r Q(L_r))$ from which the theorem follows.
\end{proof}

%The take away from this section is that if we apply Theorem~\ref{thm:flattening} with the conductor of the language then we reduce an \emph{arbitrary} regular language to a collection of flat regular languages. 

%
%\begin{theorem}
%	%\label{thm:phiequivalence}
%	Let $L = \varphi^{-1}(S) \subseteq \Sigma^{*}$ be a flat regular language recognized by $\varphi \colon \Sigma^{*} \to M$, where $M$ is the syntactic monoid. Deciding membership in $L$ is equivalent to computing $\varphi$. 
%\end{theorem}
%\begin{proof}
%	Let $w \in \Sigma^{*}$ be the input string. Clearly if we can compute $\varphi(w)$ then we can check whether it is in $S$ to decide whether $w \in L$. 
%	
%	In the other direction, we can decide which congruence class $w$ belongs to by testing $uwv \in L$ for various pairs $u,v \in \Sigma^{*}$ (this follows from the definition of syntactic monoid). Thus, we can always reduce computation of $\varphi$ to constantly many membership queries in $L$. 
%	
%	Moreover, when $L$ is \emph{flat}, we can reduce computation of $\varphi$ to membership queries in $L$ \emph{on strings of the same length}. To do this, we look at the first and last symbol of $w$, say $w = aw'b$. Then using the fact that $L$ is flat, for any $u$ and $v$ we can find symbols $a_u, b_v \in \Sigma$ such that $\varphi(a_u) = \varphi(ua)$ and $\varphi(b_v) = \varphi(bv)$, so that 
%	$$
%	\varphi(uwv) = \varphi(uaw'bv) = \varphi(a_u w' b_v),
%	$$ 
%	and therefore it suffices to test whether $a_u w' b_v$ is in $L$, instead of the longer string $uwv$. 
%\end{proof}
%!TEX root = ../reg-query.tex

\section{Equivalence of algebraic and regular expression definitions}
\label{sec:regex_parallels}

This appendix is devoted to proving Theorem~\ref{thm:monoid_characterizations}, which gives algebraic definitions for each class of regular languages defined by a regular expression.  Since Theorems~\ref{thm:reg_lang_finite_monoid} and \ref{thm:schutzenbergers_theorem} give characterizations for the regular and star-free languages, respectively, we focus on the degenerate and trivial languages.
\begin{proposition}
A language is recognized by morphism $\varphi$ such that $|\varphi(\Sigma^+)| = 1$ iff it is degenerate. 
\end{proposition}
\begin{proof}
Recall that there are only four degenerate languages: $\varnothing$, $\varepsilon$, $\Sigma^{*}$, or $\Sigma^+$. First, we claim that the morphism $\varphi \colon \Sigma^* \to M_L$ onto the syntactic monoid of each language is such that $|\varphi(\Sigma^+)| = 1$.  This calculation is straightforward, and we leave it as an exercise.  

Let language $L \subseteq \Sigma^*$ be recognized by morphism $\varphi$ such that $|\varphi(\Sigma^+)| = 1$.  For any $x, y \in \Sigma^+$, we have that $\varphi(x) = \varphi(y)$.  Therefore, $x \in L$ iff $y \in L$.  This only leaves four possible choices of languages based on whether or not $\Sigma^+ \in L$ and whether or not $\varepsilon \in L$.  These are exactly the degenerate languages.
\end{proof}

\begin{theorem}
%\label{thm:trivial_lang_rectangular_band}
A language is recognized by morphism $\varphi$ such that $\varphi(\Sigma^+)$ is a rectangular band iff it is trivial. 
\end{theorem}
\begin{proof}
Suppose first that $L$ is a regular language recognized by homomorphism $\varphi : \Sigma^* \to M$ such that $\varphi(\Sigma^+)$ is a rectangular band.  Suppose $a \in \Sigma$ belongs to $L$.  We want to show that $a \Sigma^* a$ is also in $L$.  For any $w \in \Sigma^+$, we have that $\varphi(a) = \varphi(aa) = \varphi(a w a)$, where the first equality comes from idempotence of $M$ and the second equality comes from the rectangular band property.  Therefore, if $a \in L$, then so is $a \Sigma^* a$.  Similarly, this implies that if $a w a \in L$ for $a \in \Sigma$ and $w \in \Sigma^*$, then $a \in L$ and $a \Sigma^*a \in L$.  A similar argument shows that if $a \neq b \in \Sigma$ and $a w b \in L$ for some $w \in \Sigma^*$, then $a \Sigma^* b \in L$.  Finally, membership of $\varepsilon$ is independent of $\varphi$, so it may either be in the language or not in the language.

Now suppose that $L$ is a trivial language.  Define monoid $M = (\Sigma \times \Sigma) \cup \{ (\varepsilon, \varepsilon) \}$ with operation $(a,b) \cdot (c,d) = (a, d)$ for all $a, b, c, d \in \Sigma$, and $(a,b) = (\varepsilon, \varepsilon) \cdot (a,b) = (\varepsilon, \varepsilon) \cdot (a,b)$.  Define morphism $\varphi : \Sigma^* \to M$ such that $\varphi(a) = (a,a)$ for $a \in \Sigma \cup \{ \varepsilon \}$.  Therefore, $\varphi(a w b) = (a,b)$ for $a,b \in \Sigma$ and $w \in \Sigma^*$.  Define $S \subseteq M$, such that $(a, a) \in S$ if $a \in L$, $(a,b) \in S$ if $a \Sigma^* b \in L$, and $(\varepsilon, \varepsilon) \in S$ if $\varepsilon \in L$.  By construction, we claim that $L = \varphi^{-1}(S)$, which completes the proof.
\end{proof}

%%%% Moved from formal statement

One might wonder why we needed to reference the homomorphism $\varphi$ explicitly in the definition of the degenerate and trivial languages, when the other classes only needed a characterization of the monoid itself.  In that case, each class of languages would be a \emph{variety}.  Unfortunately, such a characterization does not exist due the following theorem of Eilenberg:
\begin{theorem}[Eilenberg's Variety Theorem \cite{eilenberg:1974}]
If $V$ is a class of monoids and $\mathcal{L}$ is the class of regular languages whose syntactic monoids lie in $V$, then $V$ is a monoid variety only if $\mathcal{L}$ is a language variety.\footnote{A class of regular languages is an \emph{language variety} if it is closed under Boolean operations, left and right quotients, and inverse morphisms.  For $x \in \Sigma^*$, the \emph{left quotient} of language $L$ by $x$ is the language $x^{-1} L = \{z : xz \in L \}$.  Let $\chi \colon \Sigma_1^* \to \Sigma_2^*$ be a homomorphism, and let $L \subseteq \Sigma_1^*$ be equal to $\sum_{m \in S} \varphi^{-1}(m)$ for some subset $S$ of the syntactic monoid.  The \emph{inverse morphism} of $L$ by $\chi$ is the language $\chi^{-1} L = \sum_{m \in S} \chi^{-1} \circ \varphi^{-1} (m) \subseteq \Sigma_2^*$. } 
\end{theorem}

Consider the degenerate language $A = \Sigma^+$ and star-free language $B = \Sigma^* 1 \Sigma^*$ over alphabet $\Sigma = \{0,1\}$.  We claim that $B$ is the inverse morphism of $A$ by $\chi \colon \Sigma^* \to \Sigma^*$ such that $\chi(0) = \epsilon$, $\chi(1) = 1$.  Since $B$ is clearly nontrivial, the trivial languages are not closed under inverse morphism.  Therefore, by the Variety Theorem, the class of trivial languages is not a variety.

%!TEX root = ../reg-query.tex

\section{Worked Example}
\label{app:schutzen_example}

Let us consider the language $L \subseteq \Sigma^{*}$ (where $\Sigma = \{ a, b, c \}$) recognized by the following automaton. 
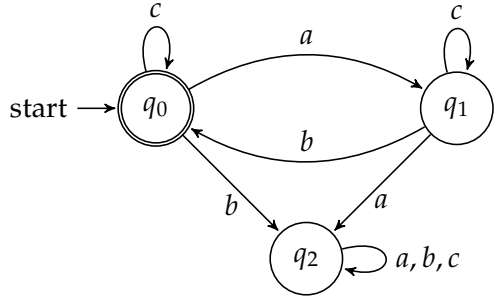
\begin{center}
\begin{tikzpicture}[->,>=stealth',shorten >=1pt,auto,node distance=2cm,
                    semithick]
  \node[initial,state,accepting] (A)                    {$q_0$};
  \node (MID)[right of=A]{};
  \node[state]         (B) [right of=MID] {$q_1$};
  \node[state]         (C) [below of=MID] {$q_2$};

  \path (A) edge [loop above] node {$c$} (A)
            edge [bend left]  node {$a$} (B)
	    edge [below] node {$b$} (C)
        (B) edge [loop above] node {$c$} (B)
            edge [bend left, above]  node {$b$} (A)
            edge [below] node {$a$} (C)
        (C) edge [loop right] node {$a,b,c$} (C);
\end{tikzpicture}
\end{center}

As this automaton is minimal, we may compute the syntactic monoid $M$ and the associated morphism $\varphi \colon \Sigma^{*} \to M$. There are six monoids elements, $M = \{ \mathbf{1}, \mathbf{a}, \mathbf{b}, \mathbf{ab}, \mathbf{ba}, \mathbf{0} \}$, corresponding to the following equivalence classes of strings. 
\begin{align*}
\varphi^{-1}(\mathbf{1}) &= \{ \varepsilon, c, cc, ccc, \ldots \} \\
\varphi^{-1}(\mathbf{a}) &= \{ a, ac, ca, aba, acc, cac, cca, \ldots \} \\
\varphi^{-1}(\mathbf{b}) &= \{ b, bc, cb, bab, bcc, cbc, ccb, \ldots \} \\
\varphi^{-1}(\mathbf{ab}) &= \{ ab, abc, acb, cab, \ldots \} \\
\varphi^{-1}(\mathbf{ba}) &= \{ ba, bac, bca, cba, \ldots \} \\ 
\varphi^{-1}(\mathbf{0}) &= \{ aa, bb, aaa, aab, aac, abb, aca,baa, \ldots \} 
\end{align*}
For convenience, we also present the complete monoid multiplication table (Table~\ref{table:multtable}) and the set of all two-sided, left, and right ideals (Table~\ref{table:ideals}).  
\begin{table}
\begin{center}
\begin{tabular}{l|llllll}
$\cdot$ & $\mathbf{1}$ & $\mathbf{a}$ & $\mathbf{b}$ & $\mathbf{ab}$ & $\mathbf{ba}$ & $\mathbf{0}$ \\
\hline
$\mathbf{1}$ & $\mathbf{1}$ & $\mathbf{a}$ & $\mathbf{b}$ & $\mathbf{ab}$ & $\mathbf{ba}$ & $\mathbf{0}$ \\
$\mathbf{a}$ & $\mathbf{a}$ & $\mathbf{0}$ & $\mathbf{ab}$ & $\mathbf{0}$ & $\mathbf{aba}$ & $\mathbf{0}$ \\ 
$\mathbf{b}$ & $\mathbf{b}$ & $\mathbf{ba}$ & $\mathbf{0}$ & $\mathbf{b}$ & $\mathbf{0}$ & $\mathbf{0}$ \\
$\mathbf{ab}$ & $\mathbf{ab}$ & $\mathbf{a}$ & $\mathbf{0}$ & $\mathbf{ab}$ & $\mathbf{0}$ & $\mathbf{0}$ \\
$\mathbf{ba}$ & $\mathbf{ba}$ & $\mathbf{0}$ & $\mathbf{b}$ & $\mathbf{0}$ & $\mathbf{ba}$ & $\mathbf{0}$ \\
$\mathbf{0}$ & $\mathbf{0}$ & $\mathbf{0}$ & $\mathbf{0}$ & $\mathbf{0}$ & $\mathbf{0}$ & $\mathbf{0}$ \\
\end{tabular} 
\end{center}
\caption{Monoid multiplication table for the example.}
\label{table:multtable}
\end{table}

\begin{table} 
Two-sided ideals:
\begin{align*}
M \mathbf{1} M &= \{ \mathbf{1}, \mathbf{a}, \mathbf{b}, \mathbf{ab}, \mathbf{ba}, \mathbf{0} \} \\
M \mathbf{a} M = M \mathbf{b} M = M \mathbf{ab} M = M \mathbf{ba} M &= \{ \mathbf{a}, \mathbf{b}, \mathbf{ab}, \mathbf{ba}, \mathbf{0} \} \\
M \mathbf{0} M &= \{ \mathbf{0} \}
\end{align*}
Left ideals:
\begin{align*}
M \mathbf{1} &= \{ \mathbf{1}, \mathbf{a}, \mathbf{b}, \mathbf{ab}, \mathbf{ba}, \mathbf{0} \} \\
M \mathbf{a} = M \mathbf{ba} &= \{ \mathbf{a}, \mathbf{ba}, \mathbf{0} \} \\
M \mathbf{b} = M \mathbf{ab} &= \{ \mathbf{b}, \mathbf{ab}, \mathbf{0} \} \\
M \mathbf{0} &= \{ \mathbf{0} \}
\end{align*}
Right ideals:
\begin{align*}
\mathbf{1} M &= \{ \mathbf{1}, \mathbf{a}, \mathbf{b}, \mathbf{ab}, \mathbf{ba}, \mathbf{0} \} \\
\mathbf{a} M = \mathbf{ab} M &= \{ \mathbf{a}, \mathbf{ab}, \mathbf{0} \} \\
\mathbf{b} M = \mathbf{ba} M &= \{ \mathbf{b}, \mathbf{ba}, \mathbf{0} \} \\
\mathbf{0} M &= \{ \mathbf{0} \}
\end{align*}
\caption{Ideals of the example monoid $M$.}
\label{table:ideals}
\end{table}

Next, observe that $L = \varphi^{-1}(\mathbf{1}) \cup \varphi^{-1}(\mathbf{ab})$, so we need to consider both $\varphi^{-1}(\mathbf{1})$ and $\varphi^{-1}(\mathbf{ab})$, plus any languages that arise from recursion. The first language, $\varphi^{-1}(\mathbf{1})$, is actually a base case in our induction. We see that $\{ a \in \Sigma : \varphi(a) = \mathbf{1} \} = \{ c \}$, so $\varphi^{-1}(\mathbf{1}) = c^{*}$.

For the second language, $\varphi^{-1}(\mathbf{ab})$, we need to apply Theorem~\ref{thm:schutzdecomp}. Let's compute the relevant sets $E$, $F$, and $G$.
\begin{align*}
E &= \{ (r,a) \in M \times \Sigma : r \varphi(a) M = \mathbf{ab}M, rM \neq \mathbf{ab}M \} \\
&= \{ (\mathbf{1},a) \}, \\
F &= \{ (a,r) \in \Sigma \times M : M \varphi(a) r = M\mathbf{ab}, Mr \neq M\mathbf{ab} \} \\
&= \{ (b, \mathbf{1}) \}, \\
G &= \{ (a,r,b) \in \Sigma \times M \times \Sigma : \mathbf{ab} \in (M \varphi(a) r M \cap M r \varphi(b) M) \backslash M \varphi(a) r \varphi(b) M \} \\
&= \{ (a, \mathbf{1}, a), (b, \mathbf{1}, b) \}.
\end{align*}
Let us argue that we have correctly identified all elements of $E$, $F$, and $G$.  For the set $E$, we must find $r$ such that $\mathbf{ab}M \subsetneq rM$.  However, there is only one right ideal that strictly contains $\mathbf{ab} M$, namely $\mathbf{1} M = M$. Since $\varphi(c)M \neq \mathbf{ab} M$ and $\varphi(b)M \neq \mathbf{ab} M$, the only choice left is $\varphi(a)M = \mathbf{a}M = \mathbf{ab}M$. The argument is identical for $F$, but with the left ideals. Finally, $\mathbf{ab}$ is in all (two-sided) ideals except $\mathbf{0}M = \{ \mathbf{0} \}$, so it must be that $\varphi(a) r \varphi(b) = \mathbf{0}$ for any $(a,r,b) \in G$. Brute force analysis of the remaining options will verify that only the two triples listed above belong in $G$.

Now that we have $E$, $F$, and $G$, the next step is to write down the relevant languages. 
\begin{align*}
U \Sigma^{*} &= \varphi^{-1}(1) a \Sigma^{*} = c^{*} a \Sigma^{*} \\
\Sigma^{*} V &= \Sigma^{*} b \varphi^{-1}(1) = \Sigma^{*} b c^{*} \\
\Sigma^{*} W \Sigma^{*} &= \Sigma^{*} a \varphi^{-1}(1) a \Sigma^{*} \cup \Sigma^{*} b \varphi^{-1}(1) b \Sigma^{*} \\
&= \Sigma^{*} a c^{*} a \Sigma^{*} \cup \Sigma^{*} b c^{*} b \Sigma^{*} \\
C &= \{ a \in \Sigma : \mathbf{ab} \notin M \varphi(a) M \} = \varnothing 
\end{align*}
In other words, $U \Sigma^{*}$ checks that there is an $a$ before the first occurrence of $b$. Similarly, $\Sigma^{*} V $ checks that there is a $b$ after all occurrences of $a$. The language $C$ is trivial, because every letter puts us in an ideal containing $\mathbf{ab}$. If there were some unused letter $d \in \Sigma$, then it would appear in $C$. Finally, $W$ checks for either two $a$s or two $b$s in a row, discounting the $c$s.

The query algorithm given in Theorem~\ref{thm:main} or in pseudocode as Algorithm~\ref{alg:sf} is not necessarily the most natural or efficient algorithm, but in this case it works quite well. First, the language $C$ is empty (there are no symbols which prevent $\varphi^{-1}(\mathbf{ab})$ from being in the monoid), so there is no work to do there. To recognize $U \Sigma^{*} = \varphi^{-1}(\mathbf{1}) a \Sigma^{*}$, the algorithm computes the longest prefix in 
$$
\bigcup_{s \colon \mathbf{1} \in sM} \varphi^{-1}(s) = \varphi^{-1}(\mathbf{1}) = c^{*},
$$
by binary search. This requires recursively solving $\varphi^{-1}(\mathbf{1})$, but fortunately this is the base case (i.e., $\rho(\mathbf{1}) = 0$). Hence, we Grover search for a symbol which does not map to $\mathbf{1}$, i.e., a symbol other than $c$, and reject if any such symbol is found. Having found the end of the longest prefix in $\varphi^{-1}(\mathbf{1})$, we check again (in this case, unnecessarily duplicating work) that the prefix is in $\varphi^{-1}(\mathbf{1})$ and that the next symbol is $a$. If so, then there is a prefix in $U$, otherwise there is not. 

Clearly $\Sigma^{*} V$ is symmetric to $U \Sigma^{*}$ and the algorithms are identical except for a reversal, so we move on to $\Sigma^{*} W \Sigma^{*}$. Its two components, $\Sigma^{*} a \varphi^{-1}(\mathbf{1}) a \Sigma^{*}$ and $\Sigma^{*} b \varphi^{-1}(\mathbf{1}) b \Sigma^{*}$ have very similar algorithms, so we discuss only the former. The algorithm applies infix search to find an index which is preceded by $a \varphi^{-1}(\mathbf{1})$ and followed by $\varphi^{-1}(\mathbf{1}) a$, since the only way to split  $\varphi^{-1}(\mathbf{1})$ is into $\varphi^{-1}(\mathbf{1}) \varphi^{-1}(\mathbf{1})$. Checking that the index is followed by $\varphi^{-1}(\mathbf{1}) a \Sigma^{*}$ is exactly the same algorithm as for $U \Sigma^{*}$ above, and $\Sigma^{*} a \varphi^{-1}(\mathbf{1})$ is just the reverse of that, so there is nothing new here.

%!TEX root = ../reg-query.tex

\section{Applications}
\label{sec:more_applications}

We present a few concrete instances where our main result implies surprising or novel quantum query algorithms.

\subsection{Addition}

Chandra, Fortune, and Lipton \cite{chandra:1983} observed that binary addition can be described by a monoid product. Specifically, a product over a monoid $M$ with elements $\{ S, R, P \}$ (\emph{set}, \emph{reset}, \emph{propagate}) satisfying 
\begin{align*}
xS &= S, & xR &= R, & xP &= x,
\end{align*}
for all $x \in M$. The idea is that given two $n$-bit numbers, we map each column to a monoid element (i.e., $00 \mapsto R$, $01, 10 \mapsto P$, $11 \mapsto S$) and then the prefix product to a particular column (starting from the least significant column, so perhaps \emph{suffix} product is more appropriate) indicates whether there is a carry in the next column ($R, P \implies \textrm{no carry}$, $S \implies \textrm{carry}$). Chandra et al.\ show that there are $\AC^{0}$ circuits for computing all prefix products, and thus binary addition can be computed in $\AC^{0}$.

Since the monoid is aperiodic, our result implies that the product of any prefix can be computed with $\tilde{O}(\sqrt{n})$ queries to the input, and therefore any particular output bit of a binary addition can be computed in the same number of queries. Similarly, the regular language accepting triples of binary numbers (represented a column at a time) such that the first two sum to the third is star free (the monoid is essentially $M$, adjoin a \emph{zero element} $\bot$ which arises when the string is inconsistent with any valid addition). This implies that addition can be checked in $\tilde{O}(\sqrt{n})$ quantum queries. Unfortunately, we cannot \emph{construct} the sum in $\tilde{O}(\sqrt{n})$ queries for information theoretic reasons: if one of the summands is zero then the sum is exactly the other summand, which we should not be able to reconstruct in fewer than $\Omega(n)$ queries.

Furthermore, we can extend these results to the addition over any base $k$, for an integer $k \geq 2$. In fact, we use the exact same monoid. For example, in decimal, if sum of the digits in a column is more than $9$, then a carry will be created. If the sum of the digits is less than $9$, then even if there is an incoming carry, there will be no outgoing carry. And if the sum of digits is exactly $9$, then a carry will propagate.

\subsection{Length-2 Word Break}

%The word break problem is often given to students as a dynamic programming exercise. 
\begin{problem}[Word Break Problem]
Given a finite dictionary of strings $D \subseteq \Sigma^{*}$ and a string $w \in \Sigma^{*}$, decide whether $w \in D^{*}$. That is, can $w$ be written as a concatenation of words in $D$?
\end{problem}
There exists a straightforward dynamic program (DP) which solves this problem in polynomial time.  Faster solutions exist (e.g., \cite{backurs16}), but still heavily rely on DP.  Since DP is sometimes claimed to be incompatible with quantum speedups \cite{ambainis18}, we find it surprising that our result gives a speedup on the following (limited) special case of the word break problem.
\begin{theorem}
Fix a dictionary $D \subseteq \Sigma \cup \Sigma^{2}$ containing strings of length $1$ or $2$. Given a string $w \in \Sigma^{*}$, there is a $\tilde{O}(\sqrt{n})$ query algorithm to decide whether $w \in D^{*}$.
\end{theorem}

The result follows from a lemma characterizing the syntactic monoids of such languages. 

\begin{lemma}
	Let $D \subseteq \Sigma^{*}$ be a set of strings of length at most $2$. Let $M$ be the \emph{flattened} syntactic monoid of $D^{*}$. For any $m \in M$, we show that $m^2 = m^3$. It follows that $M$ is aperiodic.  
\end{lemma}
\begin{proof}
	It is clear that the identity element $1 \in M$ has the property that $1^2 = 1^3$. For any other $m \in M$, $m \neq 1$, we can find a string $y \in \Sigma^{*}$ which $\varphi$ maps to $m$. Let $n$ be the length of $y$. We may assume $n$ is even because the monoid is flat.
	
	The statement $m^{2} = m^{3}$ is equivalent to saying that for all $x, z \in \Sigma^{*}$,
	$$
	x y^{2} z \in D^{*} \iff x y^{3} z \in D^{*}.
	$$
	We will argue this by showing that for any $w \in D^{*}$ containing $y^2$, there is a substring $u$ in $y^2$, aligned to the word breaks and of length $n = |y|$. This substring can be pumped up or down, to show the $\implies$ and $\Longleftarrow$ directions respectively. 
	
	Now assume $y^2$ is contained in some concatenation of words from $D$, and consider the positions where there are word breaks. If we find two word breaks (including the endpoints of the string, but not \emph{both} endpoints since then we would pump all of $y^2$!) at the same position modulo $n$, we are done because we immediately have a pumpable substring. In particular, if there are $n+1$ word breaks within $y^{2}$, then pigeonhole principle implies there are two at same position modulo $n$. We are necessarily close to this limit since $|y^2| = 2n$, and words in $D$ have length at most $2$, so the concatenation involves at least $n$ words.
	
	Let us do the math more carefully. Suppose we have a concatenation of words in with $n$ word breaks (i.e., $n+1$ words) at $n$ different positions modulo $n$. Since $n$ is even, there must be breaks at both odd and even positions. It follows that at least one of the words in the concatenation has length $1$, so the entire concatenation has length at most $2n+1$. This is just short enough that $y^2$ and the concatenation must share an endpoint. This endpoint plus the $n$ word breaks already in $y^2$ give us $n+1$ positions to apply the pigeonhole argument from before, finishing the proof. 
	
\end{proof}

We also note that this result is tight; if the dictionary contains even a single word of length $3$ or more, the query complexity may be $\Omega(n)$. For example, consider $D := \{ 0, 11, 101 \}^{*}$, and note that the parity of a bit string $x_1 \cdots x_n$ can be decided by testing whether $1 x_1 1 1 x_2 1 \cdots 1 x_n 1$ is in $D^{*}$.

\subsection{Grid Problems}

There are many instances of problems on grids which turn into regular languages if one of the dimensions is restricted to be constant. For example, $3$-colorability is $\NP$-complete for $4$-regular planar graphs \cite{garey:1974}, and such graphs may be embedded into the grid with rectilinear edges \cite{valiant:1981}. However, if one dimension of the grid is constant size then the problem becomes regular under a suitable encoding. 

In this section, we consider a grid problem such that the constant-height restriction is \emph{star free}. This leads to an efficient $\tilde{O}(\sqrt{n})$ quantum query algorithm, which is otherwise difficult to see.
\begin{problem}[Grid Path Problem]
Given an $m \times n$ grid of cells, some of which are impassable, decide whether there is a path from the bottom left corner to the top right corner. 
\end{problem} 

For constant $m$, let 
$$
L = \{ w \in (\{ 0, 1 \}^{m})^{*} : \text{grid represented by $w$ contains a path} \}.
$$
be the language of grids which have a path from the lower left corner to the top right corner.  First, consider a \emph{monotone} version of the grid path problem in which the path is only allowed to go up or to the right at each step.  In this case, there is a straightforward first order logic characterization of this language, in which the existential quantifiers are used to guess the finitely-many positions at which the path's $y$-coordinate increases.

Such a direct characterization will not suffice for the language $L$ since there is no succinct way to describe a general path.  Instead, we appeal to a more sophisticated approach of Hansen et al.\ based on a monoid which recognizes this language \cite{hksst:2014}.  Roughly speaking, the monoid elements describe sets of compatible paths between the ends of a grid.  Thus, by multiplying the monoid elements corresponding to each column of the grid, one can determine membership in $L$.  Hansen el al.\ show that the monoid is aperiodic, which immediately gives a faster quantum query algorithm using our classification:

\begin{corollary}[Combining \cite{hksst:2014} with star-free algorithm]
$Q(L) = \tilde{O}(\sqrt n)$.
\end{corollary}

In fact, the monoid elements keep track of multiple disjoint paths through the grid (which is necessary if the path backtracks through a particular section of the grid), so one can decide whether there exist $O(1)$ disjoint paths through the grid.

\section{Context-Free Technical Lemmas}
\label{sec:contextfreeproofs}

In this section we provide proofs for the two main technical lemmas in Section~\ref{sec:cfl}.  

\begin{replemma}{lem:cflcounting}
	Let $K \subseteq \{ 0, 1, \hashtag_1, \hashtag_2, \$ \}^*$ be the language such that
	\begin{itemize}[itemsep=0pt]
		\item if $x \in K$, then $x$ ends with $\$ y \hashtag_1$, and
		\item for all $n \geq 6$, there is an $x \in K$ ending in $\$ y \hashtag_1$,
	\end{itemize}
	where $y$ is the binary representation of $|x|$. Then, $\overline{K}$ is context free, and $Q(K) = O(\sqrt{n})$.  
\end{replemma}
\begin{proof}
Let $K_1$ be the language over $\Sigma := \{ 0, 1, \hashtag_1, \hashtag_2, \$ \}$ containing all strings which
\begin{itemize}[itemsep=0pt]
	\item start with $\hashtag_1 \hashtag_a$ or $\hashtag_2 \$ \hashtag_a$, 
	\item end with $\hashtag_1$, 
	\item match $((\hashtag_1|\hashtag_2)\$^{*})^{*} (0|1)^{*} \hashtag_1$, and
	\item contain no substring of the form $\hashtag_a (0|1|\$)^{i} \hashtag_b (0|1|\$)^{j} \hashtag_c$ such that $\frac{2(i+1)}{a} \neq \frac{j+1}{b}$ where $a, b, c \in \{ 1, 2\}$ and $i, j$ are integers.
\end{itemize}
Let us show that $\overline{K_1}$ is context free as a first step to constructing $K$.  We claim there is a context-free language which accepts strings containing a substring of the form $\hashtag_a (0|1|\$)^{i} \hashtag_b (0|1|\$)^{j} \hashtag_c$. Indeed, it is easy to describe the pushdown automaton: nondeterministically guess the position of $\hashtag_a$, push symbols onto the stack as we read the input, read $\hashtag_b$ and pop symbols off the stack at a ratio of $1$ stack symbol for each $\frac{2b}{a}$ input symbols. With some attention to detail, the PDA will be able to decide whether $\frac{2(i+1)}{a} = \frac{j+1}{b}$, and accept if it does not.  Since the first three conditions define a regular language, the entire language $\overline{K_1}$ is context free.

The conditions above imply that any string $z \in K_1$ is of the form
$$
\hashtag_{a_0} \$^{*} \hashtag_{a_1} \$^{*} \cdots \$^{*} \hashtag_{a_{k-1}} \$^{*} b_{\ell-1} \cdots b_1 b_0 \hashtag_1,
$$ 
where $a_0, \ldots, a_{k-1} \in \{ 1, 2\}$ and $b_0, \ldots, b_{\ell-1} \in \{ 0, 1 \}$. 
Let $d_i$ be the distance (measured by the difference in indices) between $\hashtag_{a_{i}}$ and $\hashtag_{a_{i+1}}$, for all $i=0, \ldots, k-2$. Let $d_{k-1}$ be the distance from $\hashtag_{a_{k-1}}$ to the final $\hashtag_1$. Since $\hashtag_{a_0}$ is the first symbol and $\hashtag_{1}$ is the last, it follows that $|z| = 1 + \sum_{i} d_i$. 

Since strings in $K_1$ start with $\hashtag_1 \hashtag_a$ or $\hashtag_2 \$ \hashtag_a$, we have $d_0 = a_{0}$. We also have a condition on any three consecutive $\hashtag_{a_i}$ which translates into $2 \frac{d_i}{a_i} = \frac{d_{i+1}}{a_{i+1}}$. A straightforward induction tells us that $d_i = a_i 2^{i}$ for all $i$, which means 
\begin{equation}
\label{eqn:cfl_as}
|z| = 1 + \sum_{i=0}^{k-1} a_i 2^{i}.
\end{equation}
On the other hand, we want $b_{\ell-1} \cdots b_{0}$ to be the binary number representation of $|z|$. That is, 
\begin{equation}
\label{eqn:cfl_bs}
|z| = \sum_{i=0}^{\ell-1} b_i 2^{i}.
\end{equation}
By combining (\ref{eqn:cfl_as}) and (\ref{eqn:cfl_bs}), and considering the result modulo powers of $2$, one can show that $b_i = a_i - 1$ for all $i$, and $b_{k} = 1$. Let $K_2$ be the language that accepts if the $a_i$s and $b_i$s match up as described above. Clearly $\overline{K_2}$ is context free because a PDA can easily push the $a_i$s onto the stack as it reads them, then pop off and compare as it reads the $b_i$s. 

We define $K := K_1 \cap K_2$ and note that $\overline{K} = \overline{K_1}\cup \overline{K_2}$ is context free as desired. There are strings in $K_1$ of any length $n \geq 2$, but to be in $K_2$, we also need the binary representation of $n$ to fit in $d_{k-1} - 1$. We claim the binary representations fits for all $n \geq 6$, so there exist strings of those lengths in $K$. 

Finally, we can decide whether a string $z$ of length $n$ is in $K$ in $O(\sqrt{n})$ time. First, we check if $z \in K_1$, since the length fixes the positions of $\hashtag_{a_0}$ through $\hashtag_{a_{k-1}}$ in the string. We can determine these positions $a_0, \ldots, a_{k-1} \in \{ 1, 2 \}$ from the length of the string, and check those positions in $O(\log n)$ queries. If $z$ is in $K_1$ then we check whether bits $b_k \cdots b_{0}$ at the end match the length in $O(\log n)$ queries. Finally, we check that all remaining positions are $\$$'s in $O(\sqrt{n})$ quantum queries by Grover search. 
\end{proof}

\begin{replemma}{lem:cflhistory}
Let $N$ be a $k$-tape nondeterministic Turing machine.  Define language $K_N$ which contains strings of the form
$$
C_1 \hashtag C_2^R \hashtag C_3 \ldots C_{n-1}^R \hashtag C_n
$$
where $C_1$ is a valid start configuration of $N$, $C_n$ is a valid accepting configuration, and $C_i$ to $C_{i+1}$ is a valid transition. Then, $\overline{K_N}$ is context free, and $Q(K_N) = O(\sqrt n)$.
\end{replemma}
\begin{proof}
The proof of this theorem follows from the observation that computation is local.  Let us sketch the proof.  First, we need to fix the encoding of the configuration of a Turing machine.  Many different schemes suffice, but let us assume that the encoding consists of the $k$ tapes laid out on top of each other so that each symbol of the encoding includes a $k$-tuple of the values of the $k$ tapes.  We also stipulate that one symbol on each tape is marked with the head and the current state (we can simply expand our alphabet to include these possibilities as well).  To verify that one configuration follows properly from the next, the push-down automaton for language $\overline{K_N}$ nondeterministically guesses the location on one of the tapes where a violation might occur.  It can count to the same position in the tape in the next configuration by pushing all remaining tape symbols onto the stack until the next $\hashtag$ symbol.  At this point, it can pop these symbols to count back to the same location (this is why each configuration is the reverse of the previous one).  All that remains is to check a finite set of conditions.
\end{proof}

\end{document}